\documentclass[12pt]{article}
\pdfoutput=1
\usepackage{amssymb,amsmath,amsthm}
\usepackage{graphicx, xcolor, varwidth}
\usepackage[percent]{overpic}
\usepackage{mathtools}
\usepackage{setspace}
\usepackage{cite}
\usepackage[colorlinks=true, allcolors=darkblue, urlcolor=darkblue, linktocpage=true, linkcolor=darkblue, citecolor=purple]{hyperref}
\usepackage{tikz,tkz-graph,tkz-berge}
\usepackage{float}
\usepackage[FIGTOPCAP]{subfigure}

\usepackage{pgf,pgfplots}
\pgfplotsset{compat=1.14}
\usepackage{mathrsfs}
\usetikzlibrary{arrows}
\definecolor{rvwvcq}{rgb}{0.08235294117647059,0.396078431372549,0.7529411764705882}

\colorlet{darkblue}{blue!70!black}
\colorlet{darkgreen}{green!70!black}

\newcommand{\mrm}{\mathrm}

\newcommand{\pd}{\partial}
\newcommand{\lb}{\left(}
\newcommand{\rb}{\right)}
\newcommand{\lcb}{\left\{}

\newcommand{\LL}{\mathcal{L}}

\newcommand{\SSS}{\mathcal{S}}

\newcommand{\rcb}{\right\}}
\newcommand{\lsb}{\left[}
\newcommand{\rsb}{\right]}

\newcommand\be{\begin{equation}}
\newcommand\ee{\end{equation}}
\newcommand\ba{\begin{eqnarray}}
\newcommand\ea{\end{eqnarray}}

\def\bs#1{{\color{red}#1}}

\numberwithin{equation}{section}

\newcommand{\arxiv}[1]
  {\href{http://arxiv.org/abs/#1}{arXiv:#1}}

\newtheorem{conj}{Conjecture}
\newtheorem{defn}{Definition}

\newtheorem{theorem}{Theorem}
\newtheorem{corollary}{Corollary}
\newtheorem{lemma}{Lemma}
\newtheorem{remark}{Remark}

\renewcommand{\ij}{{\langle i j \rangle}}
\newcommand{\ipj}{{\langle i j \rangle}}

\newcommand{\ipk}{{\langle i k \rangle}}

\newcommand{\jpv}{{\langle j v \rangle}}
\newcommand{\upv}{{\langle u v \rangle}}
\newcommand{\jpk}{{\langle j k \rangle}}

\newcommand{\vpi}{{\langle v i \rangle}}
\newcommand{\iip}{{\langle i i' \rangle}}
\newcommand{\jjp}{{\langle j j' \rangle}}

\usepackage{cleveref}

\graphicspath{ {./images/} }

\begin{document}

\onehalfspacing
\begin{titlepage}

\ \\
\vspace{-2.3cm}
\begin{center}

\begin{spacing}{2.3}
{\LARGE{Bounds on the Ricci curvature and solutions to the Einstein equations for weighted graphs}}
\end{spacing}

\vspace{0.5cm}
An Huang,$^{1}$ Bogdan Stoica,$^{2}$ Xuyang Xia,$^{3}$ and Xiao Zhong$^{1}$

\vspace{5mm}

{\small

\textit{
$^1$Department of Mathematics, Brandeis University, Waltham, MA 02453, USA}\\

\vspace{2mm}

\textit{
$^2$Department of Physics \& Astronomy, Northwestern University, Evanston, IL 60208, USA}\\

\vspace{2mm}

\textit{
$^3$University of Chicago Booth School of Business, Chicago, IL 60637, USA}\\

\vspace{4mm}

{\tt anhuang@brandeis.edu, bstoica@northwestern.edu, xuyang.xia@chicagobooth.edu, xzhong@brandeis.edu}

\vspace{0.3cm}
}

\end{center}

\begin{abstract}
\noindent This is a preliminary study of the equation of motion of Euclidean classical gravity on a graph \cite{Gubser:2016htz}, based on the Lin-Lu-Yau Ricci curvature on graphs \cite{LLY}. We observe that the constant edge weights configuration gives the unique solution on an infinite tree w.r.t. the asymptotically constant boundary condition. We study the minimum and maximum of the action w.r.t. certain boundary conditions, on several types of graphs of interest. We also exhibit a new class of solutions to the equations of motion on the infinite regular tree.
\end{abstract}

\vfill
\noindent nuhep-th/20-04

\end{titlepage}

\pdfbookmark[1]{\contentsname}{toc}
\setcounter{tocdepth}{3}
\tableofcontents 
\onehalfspacing
\clearpage

\section{Introduction}

\noindent Einstein's theory of General Relativity describes gravity as the curvature of a smooth manifold. Although physically and historically motivated by the need to describe gravity beyond the Newtonian limit on planetary scales and beyond, General Relativity can be applied abstractly to any Riemannian or pseudo-Riemannian manifold which is endowed with a notion of curvature. This is done by constructing an action out of the curvature, and then applying the action principle to take variation and obtain the equations of motion. The matter ``living'' on the manifold is encoded by the stress-energy tensor, which enters the action in a canonical way. The equations of motion for the manifold, known as the Einstein equations, are usually cast as the equations of motion for the manifold's metric.

In the present paper we will be concerned with a discrete version of Ricci curvature and Einstein gravitational theory. This ``discrete gravitational theory'' was first introduced in \cite{Gubser:2016htz}, building on the work of \cite{LLY} and older results \cite{bakryemery,ChungYau,Ollivier1,Ollivier2,LinYau}. In certain aspects it is most similar to two-dimensional Euclidean gravity \cite{Stoica:2018zmi}. The modern motivation for considering such discrete theories of gravity is $p$-adic, and thus they should not be thought of as naive discretizations of a continuum theory. Rather, these discrete theories of gravity are relevant for reconstructing Archimedean gravity along more complicated number-theoretic procedures \cite{Stoica:2018zmi,Huang:2020vwx}. Furthermore, in the realm of $p$-adic theories it is possible to write down sigma-models \cite{Huang:2019nog} that are analogs of Archimedean bosonic string theory; since Archimedean gravity arises from Archimedean string theory, one should expect the $p$-adic models of gravity to have a strong connection to $p$-adic string theory.

In order to reconstruct Archimedean quantities out of $p$-adic objects, it is necessary for extra adelic structure to be present across the different places, which gives the different places a kind of ``coherence.'' This structure constrains the kind of graphs one can have at the finite places, and the Archimedean space being reconstructed. For instance, in one example of this reconstruction procedure \cite{Stoica:2018zmi}, the Bruhat-Tits trees $T_p$ for all primes are paired with a hyperbolic space at the Archimedean place.

The results in this paper are independent of any kind of adelic structure, and as such we will be agnostic as to whether the various types of graphs we will discuss can be organized adelically. While for some types of graphs (such as infinite trees) this is known to be possible, our results do not depend on it, as the Lin-Lu-Yau curvature and the action derived from it apply generally to any graph. 

Let's now schematically explain the general idea behind defining Ricci curvature on graphs. In standard physics textbooks curvature is often introduced by considering the parallel transport of a vector, around a small loop in curved spacetime. As in the case of graphs there is no good notion of parallel transport, an alternative point of view must be considered. This point of view considers the transportation distance between two probability distributions sharply peaked at two points, that in the graph case are two vertices. In the Riemannian manifold case this transportation distance is sensitive to the usual Ricci curvature, and by demanding that curvature enters the transportation distance in the same way on graphs, a notion of Ricci curvature on graphs can be defined. Out of the various related curvatures that can be constructed in this way, in this paper we will only consider the Lin-Lu-Yau curvature \cite{LLY}.

Other papers which have explored related ideas are \cite{BauerLiu,ChungLinYau,Smith2014,Yamada2016,BHLY} (see also the survey article \cite{Olliviersurvey}).

\subsection{Results and outline}

We now give an outline of our paper and results. In Section \ref{secbackg} we review the construction of the Lin-Lu-Yau curvature and of associated quantities, as well as some of the properties that the curvature obeys. In Section \ref{sec333} we will prove that the constant edge length setting for a tree graph is the unique solution to the equations of motion, with a constant boundary condition. We will furthermore show that, for certain boundary conditions, there exists no solution to the equations of motion on tree graphs. Section~\ref{secbounds} is devoted to establishing bounds on the action. We will argue that, for tree graphs and hexagonal lattices, the minimum action of a finite region, subject to certain boundary conditions, is achieved by the constant edge length setting. We will furthermore show that the maximum action for complete graphs and trees is obtained by a perfect matching edge length setting, again subject to boundary conditions. We will also exhibit an upper bound on the action, for arbitrary graphs. Together with the lower bound on the action worked out in \cite{BHLY}, these two bounds restrict the range of the action for arbitrary finite graphs. In Section \ref{sec555} we will give some more general solutions to the equations of motion, in the case where no boundary condition is imposed. 

\subsection{Acknowledgments}
We thank Linyuan Lu for pointing out the role of perfect matchings in extremizing the action. This work was supported in part by a grant from the Brandeis University Provost Office. The work of A.H. was supported in part by a grant from the Simons Foundation in Homological Mirror Symmetry. The work of A. H., B. S., X.Y. X. and X. Z. was supported in part by a grant “physics from the primes” from Brandeis University Provost Office. B. S. was supported in part by the U.S. Department of Energy under grant DE-SC-0009987, and by the Simons Foundation through the It from Qubit Simons Collaboration on Quantum Fields, Gravity and Information. B. S. thanks Northwestern University Amplitudes and Insight Group and Weinberg College for support.

\section{Background and definitions}
\label{secbackg}

In this section, we introduce the notions of geodesic path length, Wasserstein cost, curvature of edges and action for graphs. These definitions are consistent with previous works \cite{Gubser:2016htz,LLY}.

Our setup is as follows. For any graph $G$, let $V(G)$ denote the set of vertices and $E(G)$ the set of edges, We assume throughout the paper that $G$ is a connected graph with no parallel edges. Furthermore, in order to discuss curvature, it is necessary to equip the edges with ``weights,'' or ``lengths,'' that is we consider a length function
\be
\ell: E(G) \to \mathbb{R}_{\geq0}.
\ee
We denote vertices in $V(G)$ by lowercase letters $i,j,k\dots$, and the notation $i\sim j$ indicates that $i$ and $j$ are neighbors in the graph, that is there exists an edge between them. We denote the edge between $i$ and $j$ by $\ipj$.

\begin{defn}
Let $\ell:E(G)\to \mathbb{R}_{\geq0}$ be a length function. The geodesic edge length $P_\ipj$ is defined as the length of the shortest path from $i$ to $j$, that is
\be
P_\ipj \coloneqq \min \sum_\alpha \ell\lb e_{\alpha} \rb,
\ee
where the edges $\{e_{\alpha}\}$ form a path from $i$ to $j$.
\end{defn}

To define the Wasserstein cost we first have to introduce a certain probability distribution on the vertices of the graph. 
\begin{defn}
\label{hereisdef2}
	For $i\in V(G)$ and $t\in\mathbb{R}$, $0<t<1$ a parameter, we introduce the probability distribution $D_{t,i}:V(G)\to~\mathbb{R}_{\geq0}$~as 
	\begin{align}
	\label{eqdefDti}
D_{t,i}(i')\coloneqq\begin{cases}
		\frac{P^{-2}_\iip}{d_i}t & \text{if}\quad i' \sim i \\
		1 - t  & \text{if}\quad i'=i \\
		0 & \text{all other vertices}
	\end{cases},
	\end{align}
where we have used the notation
	\begin{align}
        \label{hereiscdi}
		d_i &\coloneqq \sum_{i' \sim i} \frac{1}{P^2_\iip}, \\
	    \label{hereiscdi2}
	    c_i &\coloneqq \sum_{i' \sim i} \frac{1}{P_\iip}.
	\end{align}
\end{defn}

We will sometimes write $D_{i}$ instead of $D_{t,i}$ when no confusion can arise.

\begin{remark}
Definition \ref{hereisdef2} is a generalization of the probability distribution introduced in \cite{LLY}, and a cousin of the probability distribution defined in \cite{Gubser:2016htz}. In that paper, $P_\iip$ was the length of edge $\iip$. In our case, $P_\iip$ is the geodesic edge length between vertices $i$ and $i'$, which can be equal to, or less than, the edge length of $\iip$. Our definition coincides with that in \cite{Gubser:2016htz} when the graph is a tree.
\end{remark}

Now we are in a position to introduce the Wasserstein or transportation cost.

\begin{defn}
	The neighbor transportation cost of a certain amount of probability~$q$, from vertex~$i$ to a neighboring vertex $j$, is defined as $q P_\ipj$. The Wasserstein or transportation cost between two probability distributions $D_{t,i}$, $D_{t,j}$, with $i$ and $j$ adjacent vertices, is defined as the minimum sum of neighbor transportation costs, such that probability distribution $D_{t,i}$ becomes $D_{t,j}$. We denote this cost as $W_\ipj(t)$.
\end{defn}

Having the definition of the Wasserstein cost, we can proceed to define the curvature of an edge.
\begin{defn}
	The curvature $K_\ipj$ of an edge $\ipj$ on a graph is given by 
	\be 
	\label{eq28}
	K_\ipj \coloneqq \lim_{t \to 0} \frac{1}{t} \lb 1 - \frac{W_\ipj(t)}{P_\ipj} \rb,
	\ee
	and we also introduce
	\be
	\label{eq28Kt}
	K_\ipj(t) \coloneqq  1 - \frac{W_\ipj(t)}{P_\ipj}.
	\ee
\end{defn}

The curvature in Eq. \eqref{eq28} is well-defined, in that the limit exists. This is encapsulated by the following two lemmas. 
    
    \begin{lemma}
        \label{lemma1}
        Curvature $K_\ipj(t)$ is concave in $t \in [0,1]$.
    \end{lemma}
    
    Lemma~\ref{lemma1} has been proven in \cite{LLY}, so we omit the proof.

    Lemma \ref{lemma2} below is an extension of Lemma 2.2 in \cite{LLY}. 
    \begin{lemma}
    \label{lemma2}
        For any $t\in [0,1]$, any two vertices $i,j$, we have
        \be
        K_\ipj(t)  \leq \frac{t}{P_\ipj}\lb\frac{c_i}{d_i} + \frac{c_j}{d_j}\rb.
        \ee
    \end{lemma}
    \begin{proof}
        Denote by $W(D_1 , D_2)$ the transportation cost between two distributions $D_{1,2}$. Define a delta distribution at vertex $i$ as
        \be
        \Delta_i(k) =\begin{cases}
        1 \quad \mrm{if}\ k = i\\
        0 \quad \mrm{otherwise}
        \end{cases}.
        \ee
        We have 
        \be
        W(\Delta_i,\Delta_j) \leq W(\Delta_i,D_{t,i}) + W(D_{t,i},D_{t,j}) + W(D_{t,j},\Delta_j),
        \ee
        so that 
        \ba
        W(D_{t,i},D_{j,t}) \geq W(\Delta_i , \Delta_j) - W(\Delta_i,D_{t,i}) - W(D_{t,j},\Delta_j)
        \\ = P_\ipj - \sum_{k \sim i} t \frac{P^{-2}_\ipk}{d_i}P_\ipk - \sum_{k \sim j} t \frac{P^{-2}_\jpk}{d_j}P_\jpk.
        \ea
        Thus we have shown
        \be
        K_\ipj(t)  \leq \frac{t}{P_\ipj}\lb\frac{c_i}{d_i} + \frac{c_j}{d_j}\rb.
        \ee

    \end{proof}
    From Lemma \ref{lemma2}, the curvature $K_\ipj(t)$ has an upper bound. Together with Lemma \ref{lemma1}, we obtain that
    \be
    K_\ipj \coloneqq \lim_{t \to 0} \frac{1}{t} \lb 1 - \frac{W_\ipj(t)}{P_\ipj} \rb
    \ee
    exists.
    
Furthermore, Lemma \ref{lemma2} implies that the action defined by summing curvature $K_\ipj$ on the edges of a graph is bounded from above (see Section \ref{secgenmaxbound}).
\begin{defn}
	The action of a graph is given by 
	\be
	\label{Sis}
	S \coloneqq \sum_{\ipj \in E\lb G \rb} K_\ipj,
	\ee 
	where $E(G)$ is the set of all the edges in graph $G$. 
\end{defn}
Definition \eqref{Sis} is motivated from General Relativity, where the action of a manifold is the integral of the Ricci scalar. We next introduce the equation of motion on graphs, also adopted from physics. 
\begin{defn}
	The equation of motion (EoM) for an edge $\ipj$ is defined as
	\be 
	\frac{\delta S}{\delta P_\ipj} = 0.
	\ee 
\end{defn}
\begin{defn}
For a graph $G$, a solution to the EoM is an edge length setting given by the function $\ell\lb\ipj\rb$ for all $\ipj\in E(G)$, such that the EoM holds for every edge in the graph.
\end{defn}
\begin{remark}
The equation of motion is ill-defined on a graph if the action is not differentiable.
\end{remark}

\begin{remark}
In cases when the equation of motion is well-defined, it may not have an explicit formula.
\end{remark}

\begin{defn}
The constant edge length setting for a graph is the edge length setting that lets all the edges in the graph have the same edge length, i.e. $\ell\lb e_1\rb=\ell\lb e_2\rb$ for all $e_{1,2}\in E(G)$.
\end{defn}

Finally, let's introduce the boundary of a region of a graph, and boundary conditions. These notions are also motivated from physics, by analogy with the Einstein equations on Riemannian manifolds in Euclidean signature. 

For the two definitions below, consider a connected graph $G$, which could be infinite or finite.

\begin{defn}
	The finite region $\Sigma$ of a (possibly infinite) graph $G$ is defined as a connected subgraph of $G$, with $V(\Sigma)$ and $E(\Sigma)$ of finite cardinality.
\end{defn}

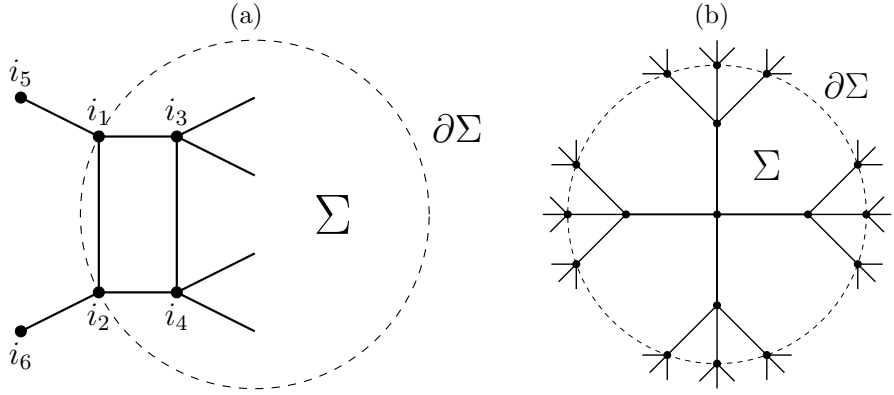
\begin{figure}[htbp!]
    \label{figsigmapdsigma}
    \centering

    \subfigure[]
    {
	\resizebox{!}{0.3\textwidth}{%
	\begin{tikzpicture}[scale=1]
	
	\filldraw[black] (0,0) circle (2pt) ;
	\filldraw[black] (1,0) circle (2pt) ;
	\filldraw[black] (0,2) circle (2pt) ;
	\filldraw[black] (1,2) circle (2pt) ;
	\filldraw[black] (-1,2.5) circle (2pt) ;
	\filldraw[black] (-1,-0.5) circle (2pt) ;
	
	\draw[black,thick] (-1,2.5) -- (0,2)--(1,2)--(1,0)--(0,0)--(0,2);
	\draw[black,thick] (-1,-0.5) --(0,0);
	
	\draw[black,thick] (1,2) --(2,2.5);
	\draw[black,thick] (1,2) --(2,1.5);
	\draw[black,thick] (1,0) --(2,0.5);
	\draw[black,thick] (1,0) --(2,-0.5);

	\draw [dashed] (2,1) circle (2.236);
	
	\filldraw[black] (3,1) circle (0pt) node[anchor=center] {\LARGE $\Sigma$};
	\filldraw[black] (4.6,2.1) circle (0pt) node[anchor=center] {\large $\partial \Sigma$};
	
	\filldraw[black] (0,0) circle (0pt) node[anchor=north] {\normalsize $i_2$};
	\filldraw[black] (0,2) circle (0pt) node[anchor=south] {\normalsize $i_1$};
	\filldraw[black] (-1,-0.5) circle (0pt) node[anchor=north] {\normalsize $i_6$};
    \filldraw[black] (-1,2.5) circle (0pt) node[anchor=south] {\normalsize $i_5$};
	\filldraw[black] (1,0) circle (0pt) node[anchor=north] {\normalsize $i_4$};
	\filldraw[black] (1,2) circle (0pt) node[anchor=south] {\normalsize $i_3$};
	
	\end{tikzpicture}
	}
	}
	\subfigure[]  
    {
	\resizebox{!}{0.3\textwidth}{
	\begin{tikzpicture}[]
	
	\filldraw[black] (0,0) circle (2pt) ;
	\filldraw[black] (1.828,0) circle (2pt) ;
	\filldraw[black] (-1.828,0) circle (2pt) ;
	\filldraw[black] (0,1.828) circle (2pt) ;
	\filldraw[black] (0,-1.828) circle (2pt) ;
	
	\filldraw[black] (3,0) circle (2pt) ;
	\filldraw[black] (-3,0) circle (2pt) ;
	\filldraw[black] (0,3) circle (2pt) ;
	\filldraw[black] (0,-3) circle (2pt) ;
	\filldraw[black] (2.828,1) circle (2pt) ;
	\filldraw[black] (-2.828,1) circle (2pt) ;
	\filldraw[black] (2.828,-1) circle (2pt) ;
	\filldraw[black] (-2.828,-1) circle (2pt) ;
	\filldraw[black] (1,2.828) circle (2pt) ;
	\filldraw[black] (1,-2.828) circle (2pt) ;
	\filldraw[black] (-1,2.828) circle (2pt) ;
	\filldraw[black] (-1,-2.828) circle (2pt) ;
	
	\draw[black,thick] (-3.5,0) -- (3.5,0);
	\draw[black,thick] (0,-3.5) -- (0,3.5);
	
	\draw[black,thick] (-2.828,1) -- (-1.828,0)--(1.828,0)--(2.828,1);
	\draw[black,thick] (-2.828,-1) -- (-1.828,0)--(1.828,0)--(2.828,-1);
	
	\draw[black,thick] (1,-2.828) -- (0,-1.828)--(0,1.828)--(1,2.828);
	\draw[black,thick] (-1,-2.828) -- (0,-1.828)--(0,1.828)--(-1,2.828);
	
	\draw[black,thick] (2.828,1) -- (3.328,1);
	\draw[black,thick] (2.828,1) -- (2.828,1.5);
	\draw[black,thick] (2.828,1) -- (3.182,1.354);
	
	\draw[black,thick] (-2.828,1) -- (-3.328,1);
	\draw[black,thick] (-2.828,1) -- (-2.828,1.5);
	\draw[black,thick] (-2.828,1) -- (-3.182,1.354);
	
	\draw[black,thick] (2.828,-1) -- (3.328,-1);
	\draw[black,thick] (2.828,-1) -- (2.828,-1.5);
	\draw[black,thick] (2.828,-1) -- (3.182,-1.354);
	
	\draw[black,thick] (-2.828,-1) -- (-3.328,-1);
	\draw[black,thick] (-2.828,-1) -- (-2.828,-1.5);
	\draw[black,thick] (-2.828,-1) -- (-3.182,-1.354);
	
	\draw[black,thick] (1,2.828) -- (1,3.328);
	\draw[black,thick] (1,2.828) -- (1.5,2.828);
	\draw[black,thick] (1,2.828) -- (1.354,3.182);
	
	\draw[black,thick] (-1,2.828) -- (-1,3.328);
	\draw[black,thick] (-1,2.828) -- (-1.5,2.828);
	\draw[black,thick] (-1,2.828) -- (-1.354,3.182);
	
	\draw[black,thick] (-1,-2.828) -- (-1,-3.328);
	\draw[black,thick] (-1,-2.828) -- (-1.5,-2.828);
	\draw[black,thick] (-1,-2.828) -- (-1.354,-3.182);
	
	\draw[black,thick] (1,-2.828) -- (1,-3.328);
	\draw[black,thick] (1,-2.828) -- (1.5,-2.828);
	\draw[black,thick] (1,-2.828) -- (1.354,-3.182);
	
	\draw[black,thick] (3,0) -- (3.354,0.354);
	\draw[black,thick] (3,0) -- (3.354,-0.354);
	
	\draw[black,thick] (-3,0) -- (-3.354,0.354);
	\draw[black,thick] (-3,0) -- (-3.354,-0.354);
	
	\draw[black,thick] (0,3) -- (0.354,3.354);
	\draw[black,thick] (0,3) -- (-0.354,3.354);
	
	\draw[black,thick] (0,-3) -- (0.354,-3.354);
	\draw[black,thick] (0,-3) -- (-0.354,-3.354);
	
	\draw [dashed] (0,0) circle (3);
	
	\filldraw[black] (1,1) circle (0pt) node[anchor=center] {\huge $\Sigma$};
	\filldraw[black] (2.6,2.6) circle (0pt) node[anchor=center] {\LARGE $\partial \Sigma$};

	\end{tikzpicture}
	
	}
	}
	\caption{Examples of regions $\Sigma$ and boundaries $\pd\Sigma$, according to Definition \ref{defnsigmadsigma}. The dashed circles represent the boundary $\pd\Sigma$. In the left panel, we have $i_{1,2,3,4}\in V\lb\Sigma\rb$, $i_{1,2}\in V\lb \pd\Sigma \rb$, $\langle i_1 i_2 \rangle, \langle i_1 i_3 \rangle, \langle i_2 i_4 \rangle, \langle i_3 i_4 \rangle \in E\lb \Sigma \rb$,  and $\langle i_1 i_2 \rangle, \langle i_1 i_5 \rangle, \langle i_2 i_6 \rangle\in E\lb \pd\Sigma \rb$. In the right panel, for each vertex $i\in V\lb \pd\Sigma \rb$, there is an unique edge from $\pd\Sigma$ to the rest of $\Sigma$.}
\end{figure}

\begin{defn}
\label{defnsigmadsigma}
	The boundary of a finite region $\Sigma$ of a graph $G$ is a pair of sets $(V(\partial \Sigma) , E(\partial \Sigma)$), such that  $V(\partial \Sigma)\subseteq V(\Sigma)$, and all vertices $v\in V\lb \Sigma \rb$ with at least one neighbor not in $V(\Sigma)$ are in $V\lb \pd\Sigma \rb$. $E(\partial \Sigma)$ is the set of edges with both endpoints in $V(\partial\Sigma)$ or one endpoint in $V(\pd\Sigma)$ and one in $V\lb G \rb-V\lb \Sigma \rb$.
\end{defn}

We informally denote the boundary pair $\lb V(\partial \Sigma) , E(\partial \Sigma) \rb$ as $\partial \Sigma$. Note that $\pd\Sigma$ is not a subgraph of $G$. Two examples of finite regions $\Sigma$ and boundaries $\pd \Sigma$ are given in Figure \ref{figsigmapdsigma}.

\section{Solutions to the equations of motion in the presence of boundary conditions}
\label{sec333}

In this section we will present results on the equations of motion when the graph $G$ is a tree. In this case, the distances $P_\ipj$ are the same as the edge lengths $\ell\lb \ipj \rb$, and the expressions derived in paper \cite{Gubser:2016htz} apply.

\subsection{Solution to the equation of motion with constant boundary condition}

We now discuss the solution to the EoM in a finite region of an infinite tree of uniform degree, with constant boundary value condition. We consider every vertex in the tree to have degree $q+1$, and we denote such a tree by $T_q$.

Let's first introduce the constant boundary value condition.
\begin{defn}[Constant boundary value condition]
	\label{cbd}
	For a finite region $\Sigma$ of $T_q$ with boundary $\pd\Sigma$, we define the constant boundary value condition as $\ell\lb \ipj \rb=c$, for all $\ipj \in E (\partial \Sigma)$ and a constant $c>0$. 
\end{defn}
Since for trees the action $S$ can be expressed explicitly and is differentiable, it is possible to obtain the expression for the equation of motion by direct computation (see~\cite{Gubser:2016htz}). We denote this equation of motion by $\mrm{tEoM}$, i.e.
\be
\label{ac_EOM_tree}
\mrm{tEoM}_\ipj \coloneqq \frac{1}{P_\ipj}\lb\frac{c^2_i}{d^2_i} + \frac{c^2_j}{d^2_j}\rb - \frac{c_i}{d_i} - \frac{c_j}{d_j} = 0,
\ee 
where $\ipj$ represents an edge in the tree, and $i,j$ are neighboring vertices.

For a tree $T_q$, the setting where all edges have the same length are solutions to the $\mrm{tEoM}$. We call this setting the \textit{constant solution}.

\begin{lemma}
	\label{ac_maxmin}
	Consider an edge $\ipj$ such that its edge length and the lengths of its adjacent edges are not all equal. If the equations of motion are obeyed, then $\ipj$ cannot be the edge of maximum or minimum edge length among its adjacent edges. 
\end{lemma}

\begin{proof}
	The equation of motion at $\ipj$ is 
	\be
	c_i d^2_j\lb P^{-1}_\ipj c_i - d_i\rb + c_j d^2_i\lb P^{-1}_\ipj c_j - d_j\rb = 0,
	\ee
	and 
	\be
	P^{-1}_\ipj c_i - d_i = \sum_{\substack{k\sim i \\ k\neq j}} \frac{1}{P_\ipk} \lb \frac{1}{P_\ipj} - \frac{1}{P_\ipk} \rb,
	\ee
	therefore $P_\ipj$ cannot be the minimum or maximum edge length among its adjacent edges.
\end{proof}

\begin{theorem}
\label{thmmm1}
	For a finite region $\Sigma \subset T_q $ obeying the constant boundary value condition, the unique solution to the equations of motion is the constant edge length setting.
\end{theorem}
\begin{proof}
	This directly follows from Lemma \ref{ac_maxmin}, by contradiction. Assume there exist two adjacent edges, $e_1,e_2\in E\lb\Sigma\rb$, with different lengths. Then we can find two other edges $e'_1$ and $e'_2$ adjacent to $e_1$ and $e_2$ respectively, such that $\max(\ell(e'_1),\ell(e'_2)) > \max(\ell(e_1),\ell(e_2))$, and $\min(\ell(e'_1),\ell(e'_2)) < \min(\ell(e_1),\ell(e_2))$. We can repeat this argument on $e'_1$, $ e_1$ and $e'_2$, $e_2$ respectively, and build a path that extends to the boundary $\partial \Sigma$. This implies there are at least two edges on the boundary with different lengths.
\end{proof}

\begin{remark}
The argument above also applies to another kind of boundary condition, which is not restricted to a finite region of the tree, but instead requires that in each direction the edge lengths asymptotically approach a constant.
\end{remark}

We explain this remark below.

\begin{defn}
    A infinite path in the tree is a path formed by infinitely many edges, i.e one that extends to infinity on both ends.
\end{defn}

\begin{defn}
    An infinite tree satisfies the asymptotically constant boundary value condition if there exists a constant $c$ such that the edge lengths along every infinite path in the tree approach $c$.
\end{defn}

\begin{theorem}
	For an infinite tree satisfying the asymptotically constant boundary value condition, the unique solution to the equations of motion \eqref{ac_EOM_tree} is the constant edge length setting.
\end{theorem}

\begin{proof}
    The proof is the same as that of Theorem \ref{thmmm1} above. 
\end{proof}

\subsection{No-go theorem  on the existence of solutions to the tree equations of motion}

In this subsection we discuss a more general boundary value problem. We will solve the equations of motion on a finite region of the tree satisfying this more general boundary condition.

First we need to review the following lemma, proven in \cite{Gubser:2016htz}.

\begin{lemma}
\label{cauchy}
For any length configuration, we have 
\be 
\label{lemcausch}
\frac{c^2_i}{d_i} \leq q+1,
\ee 
where $q+1$ represents the number of edges connected to vertex $i$. Furthermore, equality in Eq. \eqref{lemcausch} holds if and only if all edge lengths around vertex $i$ are equal.
\end{lemma}
\begin{proof}Immediate from the Cauchy-Schwarz inequality.

\end{proof}

\begin{theorem}
\label{thmmm3}
If
\be
\sum_{i \in V\lb \pd\Sigma \rb} \sum_{\substack{j\sim i \\j\in  V(\Sigma)}} \lb \frac{1}{a_\ipj} \frac{c_i}{d_i} -1 \rb <0,
\ee 
then there exist no bulk solutions. .
\end{theorem}

The intuition to proving Theorem \ref{thmmm3} comes from General Relativity, where one sometimes makes use of the so-called ``trace-reversed Einstein equations.''
\begin{proof}
Starting from Eq. \eqref{ac_EOM_tree} and summing over all edges in $\Sigma$, we have
\ba
0 &=& \sum_{\ij \in E(\Sigma)} \lsb \frac{1}{P_\ij} \lb \frac{c_i^2}{d_i^2} + \frac{c_j^2}{d_j^2} \rb - \frac{c_i}{d_i} - \frac{c_j}{d_j} \rsb \\
\label{eq39}
&=& \sum_{i \in V(\Sigma-\pd\Sigma)} \frac{c_i}{d_i}\lsb \frac{c_i^2}{d_i} - \lb q+1\rb \rsb + \sum_{i\in V(\pd\Sigma)} \frac{c_i}{d_i} \sum_{\substack{j\sim i \\j\in  V(\Sigma)}} \lb \frac{1}{P_{\langle ij \rangle}} \frac{c_i}{d_i} -1 \rb.
\ea
By Lemma \ref{cauchy} the square bracket is non-positive, which completes the proof.
\end{proof}

We  now state two particular cases of Theorem \ref{thmmm3}.

\begin{remark}
Suppose $c_i/d_i$ is constant for all $i\in V(\pd \Sigma)$, and let
\be
a^{\mrm{(bdy)}} \coloneqq \frac{c_i}{d_i}, \quad  i\in V(\pd \Sigma).
\ee
Theorem \ref{thmmm3} states that the only bulk solution obeying $P_{\ipj} \geq a^\mrm{(bdy)}$, for all $i\in V\lb \pd \Sigma \rb$, $j\sim i$, $j\in V\lb \Sigma \rb$, is given by
\be
P_{\langle kl\rangle} = a^\mrm{(bdy)},
\ee
for all $\langle kl \rangle\in E\lb \Sigma\rb$.
\end{remark}

\begin{corollary}
Suppose that for every vertex $i\in V\lb \pd\Sigma \rb$ on the boundary, there is a unique edge $\langle i i_0 \rangle$ from $\pd\Sigma$ to the rest of $\Sigma$, as in the right panel of Figure \ref{figsigmapdsigma}. Then, if
\be
P_{\langle i i_0  \rangle} \geq P_{\langle i j \rangle}
\ee
for all $i\in V\lb \pd\Sigma \rb$, $j\sim i$, $j\neq i_0$, with strict inequality for at least one edge $\ipj$, there exists no bulk solution.
\end{corollary}

\begin{proof}
In this case we have
\be
\sum_{\substack{j\sim i \\j\in  V(\Sigma)}} \lb \frac{1}{P_{\langle ij \rangle}} \frac{c_i}{d_i} -1 \rb = \frac{ \sum_{\substack{j\sim i\\ j\neq i_0}}\frac{P_{\langle ii_0 \rangle}}{P_\ipj}\lb 1 - \frac{P_{\langle ii_0 \rangle}}{P_\ipj} \rb}{ 1 + \sum_{\substack{j\sim i\\ j\neq i_0}}\frac{P_{\langle ii_0 \rangle}^2}{P_\ipj^2} },
\ee
so if $P_{\langle i i_0  \rangle} \geq P_{\langle i j \rangle}$ for all edges, and the inequality is strict for at least one edge, the second term in Eq. \eqref{eq39} is negative.
\end{proof}

\section{Bounds on the action}
\label{secbounds}

In this section we consider which edge length configurations extremize the action, for graphs which are infinite trees, or an infinite network of hexagons.

While the bulk contribution to the action is the edge sum of the Lin-Lu-Yau curvature $K_\ipj$, in general there is no one canonical prescription for the boundary term. This is because different boundary terms correspond to different boundary value problems.

We will consider two types of boundary terms:
\begin{enumerate}
\item A Gibbons-Hawking-York term, of the type considered in Section 3.2 of \cite{Gubser:2016htz}. In this formulation, the action can be written as
\be
S_\Sigma = \sum_{ \ipj \in E\lb \Sigma \rb} K_\ipj + \sum_{i \in V(\pd\Sigma)} k_i,
\ee
where $k_i$ is the $p$-adic analog of extrinsic curvature. The term $k_i$ is picked such that the equations of motion have the same expression for all edges inside $\Sigma$, assuming no contribution to the variation from the edges outside $\Sigma$. This term thus corresponds to Dirichlet boundary conditions. 
\item No boundary term, so that the action is just
\be
S_\Sigma = \sum_{\ipj \in E\lb \Sigma \rb} K_\ipj.
\ee
\end{enumerate}

\subsection{Minimum action edge length setting}
\label{sec41}

Let's first consider the question of minimizing the action on $\Sigma$, with the Gibbons-Hawking-York term. In this Section \ref{sec41} only, we will restrict ourselves to boundaries $\pd\Sigma$ such that for any vertex $i \in V\lb\pd\Sigma\rb$, there is always a unique edge (which we denote $\langle ii_0\rangle$) from $i$ into $V\lb\Sigma-\pd\Sigma\rb$, as in the second panel of Figure \ref{figsigmapdsigma}.

We have (see \cite{Gubser:2016htz} for the details)
\ba
S^\mrm{GHY}_\Sigma &=& \sum_{E\lb \Sigma \rb} K_\ipj + \sum_{V(\pd\Sigma)} k_i \\
&=& \sum_{i\in V\lb \Sigma-\pd\Sigma \rb} \sum_{j\sim i} \frac{1}{d_i P_\ipj}\lb \frac{2}{P_\ipj} - c_i \rb - \sum_{i\in V\lb \pd\Sigma \rb} \frac{c_i^2}{d_i}.
\ea
Performing the $j$ sum, this equals
\be
\label{eqhere}
S^\mrm{GHY}_\Sigma = \sum_{i\in V\lb \Sigma-\pd\Sigma \rb} \lb 2 - \frac{c_i^2}{d_i} \rb - \sum_{i\in V\lb \pd\Sigma \rb} \frac{c_i^2}{d_i}.
\ee
We have arrived at the following lemma, which applies when region $\Sigma$ obeys the constant boundary value condition in Definition \ref{cbd}.

\begin{lemma} 
For a region $\Sigma$ such that all edges in $E\lb\pd\Sigma\rb$ (i.e. not in $E(\Sigma)$, but with one endpoint in $V\lb\pd \Sigma\rb$) have the same edge length $e$, the minimum value of the action $S^\mrm{GHY}_\Sigma$ is obtained by setting all edge lengths in $\Sigma$ equal to $e$.  
\end{lemma}

\begin{proof}
Immediate from Eq. \eqref{eqhere} and Lemma \ref{cauchy}.
\end{proof}

It is also possible to show that with no boundary term in the action, the constant edge length setting is still minimizing. We encapsulate this in the following lemmas.

\begin{lemma}
\label{lema6}
The action
\be
\label{SSigmanoform}
S_\Sigma = \sum_{\ipj \in E\lb \Sigma \rb} K_\ipj
\ee
equals
\be
\label{SSigmaform}
S_{\Sigma} = \sum_{i \in V(\Sigma - \partial \Sigma)} \lb 2 - \frac{c^2_i}{d_i}\rb + \sum_{i \in V(\partial \Sigma)} \lb 2\frac{P^{-2}_\ipj}{P^{-2}_\ipj + D_i} - P^{-1}_\ipj\frac{P^{-1}_\ipj + C_i}{P^{-2}_\ipj + D_i}\rb,
\ee
where $i_0$ is the unique vertex in $V\lb\Sigma-\pd\Sigma\rb$ neighboring $i\in V\lb \pd\Sigma \rb$, and we denote
\ba
\label{capitalCi}
C_i &\coloneqq& \sum_{\substack{j\sim i\\j\neq i_0}} P^{-1}_{\ipj}, \\
\label{capitalDi}
D_i &\coloneqq& \sum_{\substack{j\sim i\\j\neq i_0}} P^{-2}_{\ipj}.
\ea
\end{lemma}

\begin{proof}
Given the explicit expression (see \cite{Gubser:2016htz})
\be
K_\ipj = \frac{2}{P^2_\ipj}\lb \frac{1}{d_i} + \frac{1}{d_j} \rb - \frac{1}{P_\ipj}\lb \frac{c_i}{d_i} + \frac{c_j}{d_j} \rb
\ee
for the tree curvature, the result follows immediately by direct computation.
\end{proof}

Given the action written in form \eqref{SSigmaform}, it is possible to obtain the result that the constant edge length setting gives the minimum action, when the edge lengths outside $\Sigma$ are fixed and all equal.

\begin{lemma}
\label{lemmashort}
For $i \in V(\partial \Sigma)$, the setting
\be
P^*_{ii_0} = \frac{1}{C_i}+ \sqrt{ \frac{1}{C_i^2} +  \frac{1}{D_i}}
\ee
minimizes the second term
\be
2\frac{P^{-2}_\ipj}{P^{-2}_\ipj + D_i} - P^{-1}_\ipj\frac{P^{-1}_\ipj + C_i}{P^{-2}_\ipj + D_i}
\ee
in the action \eqref{SSigmaform}.
\end{lemma}
\begin{proof}
By direct computation, $P_{\langle ii_0\rangle}=P_{ii_0}^*$ is the local and global minimum.
\end{proof}

\begin{theorem}
Suppose the edges in $E\lb\pd\Sigma\rb$ have lengths such that the values $C_i$ and $D_i$, as defined in Eqs. \eqref{capitalCi} -- \eqref{capitalDi}, are the same for all vertices $i\in V\lb \pd \Sigma \rb$. Then the~setting
\be
P_{\ipj} = P_{ii_0},
\ee
for all edges $\ipj \in E\lb \Sigma \rb$, minimizes the action $S_\Sigma$ in Eq. \eqref{SSigmanoform}.
\end{theorem}
\begin{proof}
From Eq. \eqref{SSigmaform}, there are two contributions to $S_\Sigma$, that is
\be
S_{\Sigma} = \sum_{i \in V(\Sigma - \partial \Sigma)} \lb 2 - \frac{c^2_i}{d_i}\rb + \sum_{i \in V(\partial \Sigma)} \lb 2\frac{P^{-2}_\ipj}{P^{-2}_\ipj + D_i} - P^{-1}_\ipj\frac{P^{-1}_\ipj + C_i}{P^{-2}_\ipj + D_i}\rb.
\ee
From Lemma \ref{cauchy}, the first term is minimized by setting all edge lengths in $E\lb\Sigma\rb$ equal, and from Lemma \ref{lemmashort} the second term is minimized by setting all edge lengths $P_{\langle ii_0\rangle}= P^*_{ii_0}$. Both minima can be achieved at the same time by setting $P_\ipj=P_{ii_0}$ for all $\ipj\in E\lb \Sigma \rb$.

\end{proof}

\subsection{Maximum action edge length setting}

We are now interested in obtaining a maximum bound on the action, both with and without the Gibbons-Hawking-York boundary term.

In this section we will impose a so-called perfect matching boundary condition. This boundary condition ensures that the perfect matching setting can be imposed for the edges in $E\lb\Sigma\rb$, and we assume that the boundary $\pd\Sigma$ is such that the perfect matching boundary condition exists. Note that in order to decide whether a boundary condition is a perfect matching boundary condition, one must have topological information on the interior of $\Sigma$.

The philosophy for requiring the perfect matching boundary condition is that, in the interior of $\Sigma$, the edge length setting maximizing the action is a perfect matching. We would like the pattern of edge lengths in the bulk to extend to the boundary, thus the perfect matching boundary condition is natural in the context of maximizing the action.

\begin{defn}
\label{matchingdefn}
    A perfect matching is a set of edges such that every vertex in the graph neighbors precisely one edge in the matching. For a fixed perfect matching, a ``perfect matching setting'' is a sequence $\{A_i\}_{i\in \mathbb{N}}$ of edge length assignments, such that the edge lengths in the perfect matching in $\{A_i\}_{i\in \mathbb{N}}$ converge to $0$ absolutely, and the edge lengths in the complement converge to nonzero (possibly distinct) values. The perfect matching setting for a graph is the edge length setting given by having the lengths of edges in the perfect matching approach $0$, and the lengths of the edges in the complement nonzero and arbitrary.
\end{defn}

\begin{defn}[Perfect matching boundary condition] 
For all edges with one endpoint in $V\lb\pd\Sigma\rb$ and the other not in $V\lb\Sigma\rb$, a ``perfect matching boundary condition'' is a sequence of edge length assignments that can be extended into a perfect matching setting for all edges with at least one endpoint in $V\lb\Sigma\rb$.
\end{defn}

For vertices $i\in V\lb\pd\Sigma\rb$, the perfect matching boundary condition assigns edge lengths such that once the assignment is extended to a perfect matching setting on $\Sigma$, precisely one edge with endpoint at $i$ has vanishing length, once the limit is taken.  A graphical representation of a perfect matching boundary condition is in Figure~\ref{figmatching}.

We will need the following lemma, which is a slight extension of Lemma \ref{cauchy}.

\begin{lemma}
\label{lemma8}
For all edge length configurations and for any vertex $i$,
\be 
1 < \frac{c^2_i}{d_i} \leq q + 1,
\ee 
where $q+1$ is the number of edges connected to vertex $i$.
\end{lemma}
\begin{proof}
Immediate.

\end{proof}

Lemma \ref{lemma8} makes the following theorem straightforward.

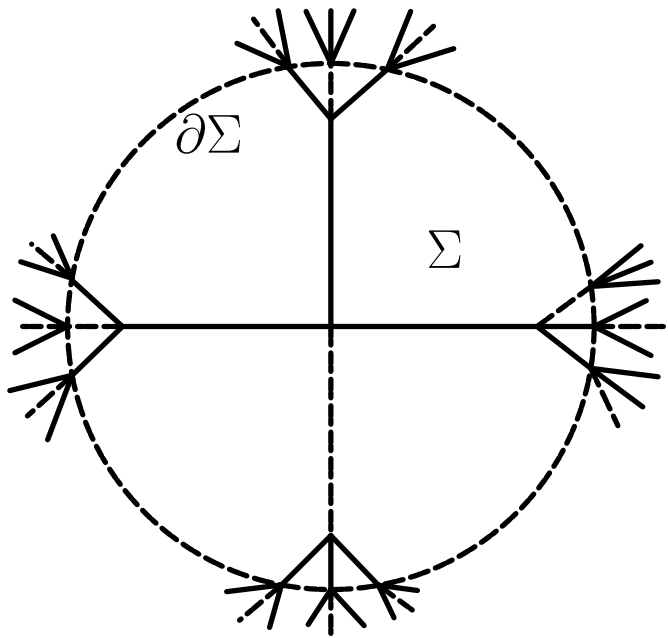
\begin{figure}[htp]
\centering
\begin{tikzpicture}[line cap=round,line join=round,>=triangle 45,x=1cm,y=1cm,scale = 0.35]
\draw [line width=2pt,dash pattern={on 4pt off 1pt on 2pt off 3pt}] (0,0) circle (10cm);
\draw [line width=2pt] (0,0)-- (-7.915694175704947,0);
\draw [line width=2pt] (0,0)-- (0,7.889539159568745);
\draw [line width=2pt,dash pattern={on 4pt off 1pt on 2pt off 3pt}] (0,0)-- (0,-7.948553011006798);
\draw [line width=2pt] (-7.915694175704947,0)-- (-9.843056467555937,1.7647207644568638);
\draw [line width=2pt,dash pattern={on 4pt off 1pt on 2pt off 3pt}] (-7.915694175704947,0)-- (-10,0);
\draw [line width=2pt] (-7.915694175704947,0)-- (-9.82669140425909,-1.853681754093862);
\draw [line width=2pt] (0,7.889539159568745)-- (-1.638339423340938,9.864879316744168);
\draw [line width=2pt,dash pattern={on 4pt off 1pt on 2pt off 3pt}] (0,7.889539159568745)-- (0,10);
\draw [line width=2pt] (0,7.889539159568745)-- (2.137192847737361,9.768951158214493);
\draw [line width=2pt] (7.83489472321005,0)-- (0,0);
\draw [line width=2pt,dash pattern={on 4pt off 1pt on 2pt off 3pt}] (7.83489472321005,0)-- (9.88367633872326,1.520835964646402);
\draw [line width=2pt] (7.83489472321005,0)-- (10,0);
\draw [line width=2pt] (7.83489472321005,0)-- (9.872641636536386,-1.5908950677301536);
\draw [line width=2pt] (0,-7.948553011006798)-- (0,-10);
\draw [line width=2pt] (0,-7.948553011006798)-- (-1.8728051897035227,-9.823064731610983);
\draw [line width=2pt] (-1.8728051897035227,-9.823064731610983)-- (-3.957793913401622,-10.079473458051754);
\draw [line width=2pt,dash pattern={on 4pt off 1pt on 2pt off 3pt}] (-1.8728051897035227,-9.823064731610983)-- (-3.468800088497865,-11.166126402282321);
\draw [line width=2pt] (-1.8728051897035227,-9.823064731610983)-- (-2.3278144970557663,-11.437789638339962);
\draw [line width=2pt] (-0.8608330223444961,-11.329124343916905)-- (0,-10);
\draw [line width=2pt,dash pattern={on 4pt off 1pt on 2pt off 3pt}] (0,-10)-- (0,-11.664998890315447);
\draw [line width=2pt] (0,-10)-- (1.263079550469801,-11.368638996434382);
\draw [line width=2pt] (1.80654640189429,-9.835465931912061)-- (0,-7.948553011006798);
\draw [line width=2pt] (1.8065464018942885,-9.83546593191206)-- (2.399125810347214,-11.07227910255332);
\draw [line width=2pt,dash pattern={on 4pt off 1pt on 2pt off 3pt}] (1.8065464018942885,-9.83546593191206)-- (3.090632229403031,-10.775919208672256);
\draw [line width=2pt] (1.8065464018942885,-9.83546593191206)-- (3.288205491990407,-10.035019473969598);
\draw [line width=2pt,dash pattern={on 4pt off 1pt on 2pt off 3pt}] (-10,0)-- (-12,0);
\draw [line width=2pt] (-10,0)-- (-12,1);
\draw [line width=2pt] (-10,0)-- (-12,-1);
\draw [line width=2pt,dash pattern={on 4pt off 1pt on 2pt off 3pt}] (-9.82669140425909,-1.853681754093862)-- (-11.524502920006352,-3.400042280686983);
\draw [line width=2pt] (-9.82669140425909,-1.853681754093862)-- (-12.193177585939345,-2.4363640856659203);
\draw [line width=2pt] (-9.82669140425909,-1.853681754093862)-- (-10.757493744377331,-4.30471976989043);
\draw [line width=2pt,dash pattern={on 4pt off 1pt on 2pt off 3pt}] (-9.83095324854881,1.8309446263739444)-- (-11.367167704492708,3.1293691631291986);
\draw [line width=2pt] (-9.83095324854881,1.8309446263739444)-- (-11.780172645216026,2.460694497196216);
\draw [line width=2pt] (-9.83095324854881,1.8309446263739444)-- (-10.541157823046068,3.4440395941564845);
\draw [line width=2pt,dash pattern={on 4pt off 1pt on 2pt off 3pt}] (-1.5819243998396173,9.874083005180383)-- (-2.9374638829076995,11.683518419812703);
\draw [line width=2pt] (-1.5819243998396173,9.874083005180383)-- (-3.5655885563671083,10.763764433675728);
\draw [line width=2pt] (-1.5819243998396173,9.874083005180383)-- (-2,12);
\draw [line width=2pt] (0,12)-- (0,10);
\draw [line width=2pt] (0,10)-- (-0.9409247422688658,11.975147732490281);
\draw [line width=2pt] (0,10)-- (0.8761502059529943,11.997580756542403);
\draw [line width=2pt,dash pattern={on 4pt off 1pt on 2pt off 3pt}] (2.137192847737361,9.768951158214493)-- (3.9494745010936713,11.45918817929149);
\draw [line width=2pt] (2.137192847737361,9.768951158214493)-- (3.0297205149566806,11.975147732490281);
\draw [line width=2pt] (2.137192847737361,9.768951158214493)-- (4.667331270761567,10.561867217206636);
\draw [line width=2pt] (9.88367633872326,1.520835964646402)-- (11.778599895284156,3.0692371837980996);
\draw [line width=2pt] (9.88367633872326,1.520835964646402)-- (12.159961304170224,2.4411125103387014);
\draw [line width=2pt] (9.88367633872326,1.520835964646402)-- (12.429157592795684,1.7008227166186964);
\draw [line width=2pt,dash pattern={on 4pt off 1pt on 2pt off 3pt}] (10,0)-- (12.63105480926478,0);
\draw [line width=2pt] (10,0)-- (11.98049711175325,0.9829659469508127);
\draw [line width=2pt] (10,0)-- (12.159961304170224,-1.1033052898964741);
\draw [line width=2pt] (9.872641636536386,-1.5908950677301536)-- (12.429157592795684,-1.9108941557728433);
\draw [line width=2pt] (9.872641636536386,-1.5908950677301536)-- (11.84589896744052,-3.054978382431033);
\draw [line width=2pt,dash pattern={on 4pt off 1pt on 2pt off 3pt}] (9.872641636536386,-1.5908950677301536)-- (10.903711957251408,-3.7728351520989167);
\begin{scriptsize}
\draw[color=black ] (-4.643450710174212,7.332513012163766) node {\LARGE $\partial \Sigma$};
\draw[color=black] (4.322608929886176,2.9409701772485084) node {\LARGE $\Sigma$};
\end{scriptsize}
\end{tikzpicture}

\caption{The perfect matching setting which maximizes $S_\Sigma^\mrm{GHY}$. The dashed edges have lengths approaching zero, and the solid edges have finite, not necessarily equal, lengths.}
\label{figmatching}
\end{figure}

\begin{theorem} Consider the Lin-Lu-Yau curvature action, together with the Gibbons-Hawking-York boundary term, that is
\be
S^\mrm{GHY}_\Sigma = \sum_{\ipj \in E\lb \Sigma \rb} K_\ipj + \sum_{ i \in V(\pd\Sigma)} k_i.
\ee

Then the perfect matching setting for all edges with endpoints in $V\lb\Sigma\rb$ maximizes $S_\Sigma^\mrm{GHY}$.
\end{theorem}

The perfect matching setting maximizing $S_\Sigma$ is in Figure \ref{figmatching}.

\begin{proof}
From Eq. \eqref{eqhere}, $S_\Sigma^\mrm{GHY}$ equals
\be
\label{SGHYhere}
S^\mrm{GHY}_\Sigma = \sum_{i\in V\lb \Sigma-\pd\Sigma \rb} \lb 2 - \frac{c_i^2}{d_i} \rb - \sum_{i\in V\lb \pd\Sigma \rb} \frac{c_i^2}{d_i}.
\ee
We thus need to minimize $\sum_{i\in V(\Sigma)} c_i^2/d_i$. From Lemma \ref{lemma8}, we have that $c_i^2/d_i>1$. For the perfect matching setting in Definition \ref{matchingdefn} one edge length adjacent to vertex $i$ approaches zero, and all other edge lengths adjacent to $i$ are finite, for all vertices $i\in V\lb\Sigma\rb$; in this limit $c_i^2/d_i\to 1$. Thus the individual contribution of each vertex $i$ to action \eqref{SGHYhere} is maximized.
\end{proof}

We now argue that the perfect matching setting maximizes the action even without the Gibbons-Hawking-York boundary term, subject to the perfect matching boundary condition, when there is a unique edge $\langle i i_0 \rangle$ between any vertex $i\in V\lb\pd\Sigma\rb$ and $V\lb \Sigma-\pd\Sigma \rb$.

From Lemma \ref{lema6}, in this case the action without the Gibbons-Hawking-York boundary term~equals 
\be
\label{eq419}
S_{\Sigma} = \sum_{i \in V(\Sigma - \partial \Sigma)} \lb 2 - \frac{c^2_i}{d_i}\rb + \sum_{i \in V(\partial \Sigma)} \lb 2\frac{P^{-2}_{\langle ii_0 \rangle}}{P^{-2}_{\langle ii_0 \rangle} + D_i} - P^{-1}_{\langle ii_0 \rangle}\frac{P^{-1}_{\langle ii_0 \rangle} + C_i}{P^{-2}_{\langle ii_0 \rangle} + D_i}\rb.
\ee 

This expression allows us to obtain an upper bound for the action $S_\Sigma$, as explained in the theorem below.

\begin{theorem} 
Consider a finite region $\Sigma$ with the perfect matching boundary condition, and a perfect matching induced on $E\lb\Sigma\rb$ by the boundary condition. Let $\LL$ be the set of all possible edge lengths in $E(\Sigma)\cup E\lb\pd\Sigma\rb$, such that the edges not in the perfect matching have lengths greater than any $\epsilon>0$. For each element in $\LL$, form an absolutely convergent sequence where the lengths of edges in $E\lb \pd\Sigma \rb$ converge to the perfect matching boundary condition, and let $S_\mrm{LIM}$ be the set of the limits of the action obtained in this manner.  Let $\mathcal{S}$ be the limit as $i\to\infty$ of the action for the perfect matching setting $\{A_i\}_{i\in\mathbb{N}}$ on all edges in $E\lb\Sigma\rb$. Then $\SSS$ is the supremum of $S_\mrm{LIM}$.
\end{theorem}

By analogy with general relativity, the restriction that $\epsilon>0$ for the edges not in the matching can be thought of as a kind of ultraviolet cutoff.

\begin{proof}
From Lemma \ref{lemma8}, the perfect matching setting maximizes the first term in Eq.~\eqref{eq419}. To extremize the second term, we note that there are two cases, depending on whether edge $\langle ii_0 \rangle$ is in the matching or not. If $\langle ii_0 \rangle$ is not in the matching, then there must be another edge neighboring $i_0$ which is in the matching and so has vanishing length; in this case we have
\be
2\frac{P^{-2}_{\langle ii_0 \rangle}}{P^{-2}_{\langle ii_0 \rangle} + D_i} - P^{-1}_{\langle ii_0 \rangle}\frac{P^{-1}_{\langle ii_0 \rangle} + C_i}{P^{-2}_{\langle ii_0 \rangle} + D_i} \to 0
\ee
for any length assignment of edge $\langle ii_0 \rangle$. If instead edge $\langle ii_0 \rangle$ is in the matching, then 
\be
2\frac{P^{-2}_{\langle ii_0 \rangle}}{P^{-2}_{\langle ii_0 \rangle} + D_i} - P^{-1}_{\langle ii_0 \rangle}\frac{P^{-1}_{\langle ii_0 \rangle} + C_i}{P^{-2}_{\langle ii_0 \rangle} + D_i} = 1.
\ee
But note that
\be
2\frac{P^{-2}_{\langle ii_0 \rangle}}{P^{-2}_{\langle ii_0 \rangle} + D_i} - P^{-1}_{\langle ii_0 \rangle}\frac{P^{-1}_{\langle ii_0 \rangle} + C_i}{P^{-2}_{\langle ii_0 \rangle} + D_i} = \frac{1-C_i P_{\langle i i_0\rangle}}{1+ D_i P_{\langle i i_0\rangle}^2} \leq 1,
\ee
therefore $1$ is the maximum value the second term can achieve in this case.
\end{proof}

\subsubsection{Maximum bound for arbitrary finite graphs}
\label{secgenmaxbound}

We now derive a maximum bound for the action on an arbitrary graph $G$. The derivation will be local, so it applies both to finite and infinite graphs. For the case discussed in the previous section, this bound will be weaker than the one derived there. 

\begin{lemma} For a finite graph $G$ the action
\label{lemmabound}
\be
S = \sum_{\ipj \in E\lb G \rb} K_\ipj
\ee
is bounded by
\be
S \leq 2 \left| E(G) \right|.
\ee
\end{lemma}
\begin{proof}
From Lemma \ref{lemma2}, by taking the limit we have
\ba
S &=& \sum_{\ipj \in E\lb G \rb} K_\ipj \leq \sum_{\ipj \in E\lb G \rb}  \frac{1}{P_\ipj}\lb \frac{c_i}{d_i} + \frac{c_j}{d_j} \rb \\
\label{eq20}
&=& \sum_{i\in V\lb G \rb} \frac{c_i^2}{d_i} \leq \sum_{i\in V\lb G\rb} \deg (i) \\
&=& 2 \left| E(G) \right|,
\ea
where in the second inequality we used Lemma \ref{cauchy}.
\end{proof}

\begin{remark}
For the complete graph $K_n$, Lemma \ref{lemmabound} gives
\be
S \leq n(n-1).
\ee
In Section \ref{subseccomplete} we will derive a stronger bound in the case of complete graphs.
\end{remark}

\begin{remark}
The result derived above is local, in the sense that the contribution of each bulk vertex is bounded from above by $c_i^2/d_i$. Thus, it applies both to finite and infinite graph, with the appropriate boundary terms.
\end{remark}

\subsection{Hexagon lattice}

\noindent We consider a graph consisting of an infinite lattice of hexagons (see Figure \ref{hexlat}), which we denote as $G$. We will prove that the constant edge length setting gives the minimum action on a finite region $\Sigma$ of $G$, given a stronger constant boundary value condition.

In this section we will only consider the action $S_\Sigma$ with no Gibbons-Hawking-York boundary term, as it not immediate how extrinsic curvature should be defined for graphs which are not trees. Furthermore our results will be off-shell, that is we do not demand the equations of motion to be satisfied.

\begin{defn}[Strong boundary value condition]
\label{defnstrong}
For a finite region $\Sigma$ in $G$, we define the strong constant boundary value condition as requiring that all edges in $E\lb\pd\Sigma\rb$ have length equal to a constant $c$. Furthermore, we require all edges that are one or two edge steps away from $E\lb\pd\Sigma\rb$ to also have lengths equal to $c$.
\end{defn}

The strong boundary condition is represented graphically in Figure \ref{hexlat}. In the rest of this section we will assume that the graph $G$ satisfies the strong boundary condition in Definition \ref{defnstrong}.

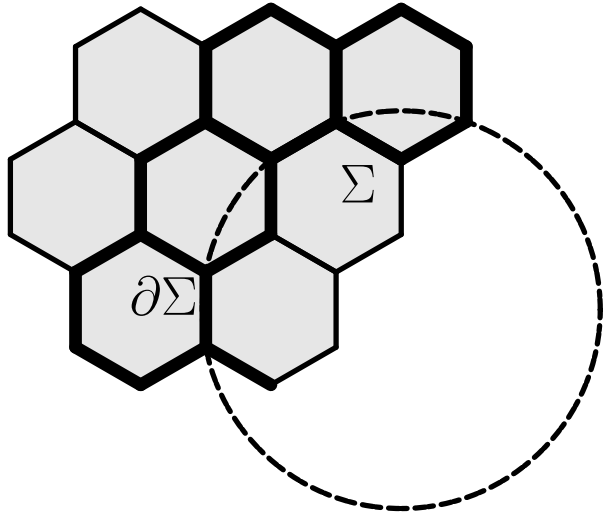
\begin{figure}[ht!]
\centering

\definecolor{zzttqq}{rgb}{0,0,0}
\definecolor{cqcqcq}{rgb}{255,255,255}
\begin{tikzpicture}[line cap=round,line join=round,>=triangle 45,x=1cm,y=1cm, scale = 0.2]
\draw [color=cqcqcq,, xstep=5cm,ystep=5cm] (-15.11331991847177,-16.522594339598935) grid (18.49267445669261,23.39402793859347);

\fill[line width=2pt,dash pattern=on 1pt off 1pt,color=zzttqq,fill=zzttqq,fill opacity=0.10000000149011612] (-5,10) -- (-5,5) -- (-0.6698729810778059,2.5) -- (3.660254037844389,5) -- (3.66025403784439,10) -- (-0.6698729810778023,12.5) -- cycle;
\fill[line width=2pt,color=zzttqq,fill=zzttqq,fill opacity=0.10000000149011612] (-5,10) -- (-0.6698729810778023,12.5) -- (-0.6698729810777992,17.5) -- (-5,20) -- (-9.330127018922194,17.5) -- (-9.330127018922198,12.5) -- cycle;
\fill[line width=2pt,color=zzttqq,fill=zzttqq,fill opacity=0.10000000149011612] (-5,5) -- (-5,10) -- (-9.330127018922194,12.5) -- (-13.660254037844389,10) -- (-13.660254037844389,5) -- (-9.330127018922198,2.5) -- cycle;
\fill[line width=2pt,color=zzttqq,fill=zzttqq,fill opacity=0.10000000149011612] (-5,5) -- (-9.330127018922198,2.5) -- (-9.330127018922201,-2.5) -- (-5,-5) -- (-0.6698729810778055,-2.5) -- (-0.669872981077801,2.5) -- cycle;
\fill[line width=2pt,color=zzttqq,fill=zzttqq,fill opacity=0.10000000149011612] (-0.6698729810778059,2.5) -- (-0.6698729810778055,-2.5) -- (3.660254037844398,-5) -- (7.990381056766601,-2.5) -- (7.990381056766601,2.5) -- (3.6602540378444006,5) -- cycle;
\fill[line width=2pt,color=zzttqq,fill=zzttqq,fill opacity=0.10000000149011612] (3.660254037844389,5) -- (7.990381056766601,2.5) -- (12.320508075688806,5) -- (12.320508075688797,10) -- (7.990381056766587,12.5) -- (3.660254037844383,10) -- cycle;
\fill[line width=2pt,color=zzttqq,fill=zzttqq,fill opacity=0.10000000149011612] (-0.6698729810777992,17.5) -- (-0.6698729810778023,12.5) -- (3.660254037844394,10) -- (7.990381056766597,12.5) -- (7.990381056766602,17.5) -- (3.6602540378444064,20) -- cycle;
\fill[line width=2pt,color=zzttqq,fill=zzttqq,fill opacity=0.10000000149011612] (7.990381056766602,17.5) -- (7.990381056766587,12.5) -- (12.320508075688746,10) -- (16.65063509461092,12.5) -- (16.650635094610933,17.5) -- (12.320508075688776,20) -- cycle;
\draw [line width=4.8pt,color=zzttqq] (-5,10)-- (-5,5);
\draw [line width=4.8pt,color=zzttqq] (-5,5)-- (-0.6698729810778059,2.5);
\draw [line width=2pt,dash pattern=on 1pt off 1pt,color=zzttqq] (-0.6698729810778059,2.5)-- (3.660254037844389,5);
\draw [line width=2pt,dash pattern=on 1pt off 1pt,color=zzttqq] (3.660254037844389,5)-- (3.66025403784439,10);
\draw [line width=5.2pt,color=zzttqq] (3.66025403784439,10)-- (-0.6698729810778023,12.5);
\draw [line width=4.8pt,color=zzttqq] (-0.6698729810778023,12.5)-- (-5,10);
\draw [line width=2pt,color=zzttqq] (-5,10)-- (-0.6698729810778023,12.5);
\draw [line width=4.8pt,color=zzttqq] (-0.6698729810778023,12.5)-- (-0.6698729810777992,17.5);
\draw [line width=2pt,color=zzttqq] (-0.6698729810777992,17.5)-- (-5,20);
\draw [line width=2pt,color=zzttqq] (-5,20)-- (-9.330127018922194,17.5);
\draw [line width=2pt,color=zzttqq] (-9.330127018922194,17.5)-- (-9.330127018922198,12.5);
\draw [line width=2pt,color=zzttqq] (-9.330127018922198,12.5)-- (-5,10);
\draw [line width=2pt,color=zzttqq] (-5,5)-- (-5,10);
\draw [line width=2pt,color=zzttqq] (-5,10)-- (-9.330127018922194,12.5);
\draw [line width=2pt,color=zzttqq] (-9.330127018922194,12.5)-- (-13.660254037844389,10);
\draw [line width=2pt,color=zzttqq] (-13.660254037844389,10)-- (-13.660254037844389,5);
\draw [line width=2pt,color=zzttqq] (-13.660254037844389,5)-- (-9.330127018922198,2.5);
\draw [line width=4.8pt,color=zzttqq] (-9.330127018922198,2.5)-- (-5,5);
\draw [line width=2pt,color=zzttqq] (-5,5)-- (-9.330127018922198,2.5);
\draw [line width=4.8pt,color=zzttqq] (-9.330127018922198,2.5)-- (-9.330127018922201,-2.5);
\draw [line width=4.8pt,color=zzttqq] (-9.330127018922201,-2.5)-- (-5,-5);
\draw [line width=4.8pt,color=zzttqq] (-5,-5)-- (-0.6698729810778055,-2.5);
\draw [line width=2pt,color=zzttqq] (-0.6698729810778055,-2.5)-- (-0.669872981077801,2.5);
\draw [line width=2pt,color=zzttqq] (-0.669872981077801,2.5)-- (-5,5);
\draw [line width=4.8pt,color=zzttqq] (-0.6698729810778059,2.5)-- (-0.6698729810778055,-2.5);
\draw [line width=4.8pt,color=zzttqq] (-0.6698729810778055,-2.5)-- (3.660254037844398,-5);
\draw [line width=2pt,color=zzttqq] (3.660254037844398,-5)-- (7.990381056766601,-2.5);
\draw [line width=2pt,color=zzttqq] (7.990381056766601,-2.5)-- (7.990381056766601,2.5);
\draw [line width=2pt,color=zzttqq] (7.990381056766601,2.5)-- (3.6602540378444006,5);
\draw [line width=4.8pt,color=zzttqq] (3.6602540378444006,5)-- (-0.6698729810778059,2.5);
\draw [line width=2pt,color=zzttqq] (3.660254037844389,5)-- (7.990381056766601,2.5);
\draw [line width=2pt,color=zzttqq] (7.990381056766601,2.5)-- (12.320508075688806,5);
\draw [line width=2pt,color=zzttqq] (12.320508075688806,5)-- (12.320508075688797,10);
\draw [line width=4.8pt,color=zzttqq] (12.320508075688797,10)-- (7.990381056766587,12.5);
\draw [line width=4.8pt,color=zzttqq] (7.990381056766587,12.5)-- (3.660254037844383,10);
\draw [line width=4.8pt,color=zzttqq] (3.660254037844383,10)-- (3.660254037844389,5);
\draw [line width=4.8pt,color=zzttqq] (-0.6698729810777992,17.5)-- (-0.6698729810778023,12.5);
\draw [line width=4.8pt,color=zzttqq] (-0.6698729810778023,12.5)-- (3.660254037844394,10);
\draw [line width=4.8pt,color=zzttqq] (3.660254037844394,10)-- (7.990381056766597,12.5);
\draw [line width=4.8pt,color=zzttqq] (7.990381056766597,12.5)-- (7.990381056766602,17.5);
\draw [line width=4.8pt,color=zzttqq] (7.990381056766602,17.5)-- (3.6602540378444064,20);
\draw [line width=4.8pt,color=zzttqq] (3.6602540378444064,20)-- (-0.6698729810777992,17.5);
\draw [line width=4.8pt,color=zzttqq] (7.990381056766602,17.5)-- (7.990381056766587,12.5);
\draw [line width=4.8pt,color=zzttqq] (7.990381056766587,12.5)-- (12.320508075688746,10);
\draw [line width=4.8pt,color=zzttqq] (12.320508075688746,10)-- (16.65063509461092,12.5);
\draw [line width=4.8pt,color=zzttqq] (16.65063509461092,12.5)-- (16.650635094610933,17.5);
\draw [line width=4.8pt,color=zzttqq] (16.650635094610933,17.5)-- (12.320508075688776,20);
\draw [line width=4.8pt,color=zzttqq] (12.320508075688776,20)-- (7.990381056766602,17.5);
\draw [line width=2pt,dash pattern={on 4pt off 1pt on 2pt off 3pt}] (12.320508075688348,3.7492016291758285E-13) circle (13.228756555322531cm);
\begin{scriptsize}
\draw[color =  zzttqq] (0.7772973313231501,8.3777747958449) node { };
\draw[color=zzttqq] (-3.480475711684101,15.904909282589754) node {};
\draw[color=zzttqq] (-7.814280416173625,8.3777747958449) node {};
\draw[color=zzttqq] (-3.480475711684101,0.9266719705823174) node {\LARGE $\partial \Sigma$};
\draw[color=zzttqq] (5.111102035812674,0.9266719705823174) node {};
\draw[color=zzttqq] (9.444906740302197,8.3777747958449) node {\LARGE $\Sigma$};
\draw[color=zzttqq] (5.111102035812674,15.904909282589754) node {};
\draw[color=zzttqq] (14.082838090720811,15.904909282589754) node {};
\end{scriptsize}
\end{tikzpicture}
\caption{The hexagon lattice graph $G$. The boundary $\pd\Sigma$ is represented dashed, and the edges to which the strong boundary condition applies are in bold. \label{hexlat}}
\end{figure}

\begin{defn}
\label{treeKdefn}
Given an edge $\ipj \in E(G)$, we define the tree-curvature on it as 
\be
\label{KTis}
K^T_\ipj \coloneqq -\frac{1}{P_\ipj}\lsb \frac{c^T_i}{d^t_i} + \frac{c^T_j}{d^T_j} - \frac{2}{P_\ipj}\lb\frac{1}{d^T_i} + \frac{1}{d^T_j}\rb\rsb,
\ee
where 
\begin{align}
c^T_i &\coloneqq \sum_{v\sim i}\ \frac{1}{P_\vpi}, \\
d^T_i &\coloneqq \sum_{v\sim i}\ \frac{1}{(P_\vpi)^2}.
\end{align}
\end{defn}

\begin{defn}
\label{defST}
Given the hexagon lattice graph $G$, we define the tree-action of a finite region $\Sigma$ as 
\be 
S_{T} \coloneqq \sum_{\ipj \in E(\Sigma)} K^T_\ipj.
\ee
\end{defn}

The intuition behind Definitions \ref{treeKdefn} and \ref{defST} is that the curvature for each edge is computed as if coming from a Wasserstein distance where the 1-Lipschitz extremization is saturated for the edge and all its edge neighbors.

\begin{lemma}
\label{lemma99}
We have
\be 
S_{T} =  \sum_{i \in V(\Sigma - \partial \Sigma)}\lb 2 - \frac{(c^T_i)^2}{d^T_i}\rb - \sum_{i \in V(\partial \Sigma) } \frac{1}{3}.
\ee
\end{lemma}
\begin{proof}
This follows from direct computation, by summing up the edge curvatures around each vertex.
\end{proof}

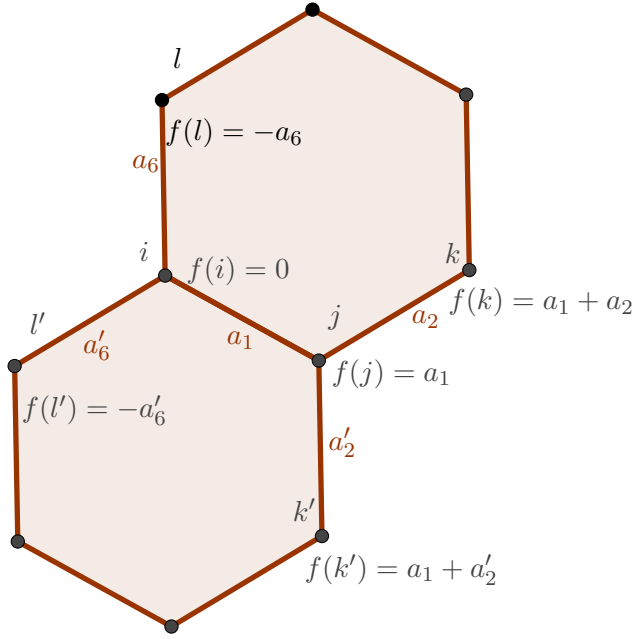
\begin{figure}[htp]
\centering

\definecolor{uuuuuu}{rgb}{0.26666666666666666,0.26666666666666666,0.26666666666666666}
\definecolor{zzttqq}{rgb}{0.6,0.2,0}
\begin{tikzpicture}[line cap=round,line join=round,>=triangle 45,x=1cm,y=1cm, scale = 0.25]

\fill[line width=2pt,color=zzttqq,fill=zzttqq,fill opacity=0.10000000149011612] (0,20) -- (-8.003596254831324,15.185221181603516) -- (-7.8356736119124095,5.846514094087353) -- (0.33584528583783135,1.3225858249676747) -- (8.339441540669156,6.137364643364156) -- (8.171518897750246,15.476071730880317) -- cycle;
\fill[line width=2pt,color=zzttqq,fill=zzttqq,fill opacity=0.10000000149011612] (0.33584528583783135,1.3225858249676747) -- (-7.8356736119124095,5.846514094087353) -- (-15.839269866743738,1.031735275690869) -- (-15.671347223824824,-8.306971811825295) -- (-7.499828326074586,-12.830900080944975) -- (0.5037679287567434,-8.016121262548495) -- cycle;
\draw [line width=2pt,color=zzttqq] (0,20)-- (-8.003596254831324,15.185221181603516);
\draw [line width=2pt,color=zzttqq] (-8.003596254831324,15.185221181603516)-- (-7.8356736119124095,5.846514094087353);
\draw [line width=2pt,color=zzttqq] (-7.8356736119124095,5.846514094087353)-- (0.33584528583783135,1.3225858249676747);
\draw [line width=2pt,color=zzttqq] (0.33584528583783135,1.3225858249676747)-- (8.339441540669156,6.137364643364156);
\draw [line width=2pt,color=zzttqq] (8.339441540669156,6.137364643364156)-- (8.171518897750246,15.476071730880317);
\draw [line width=2pt,color=zzttqq] (8.171518897750246,15.476071730880317)-- (0,20);
\draw [line width=2pt,color=zzttqq] (0.33584528583783135,1.3225858249676747)-- (-7.8356736119124095,5.846514094087353);
\draw [line width=2pt,color=zzttqq] (-7.8356736119124095,5.846514094087353)-- (-15.839269866743738,1.031735275690869);
\draw [line width=2pt,color=zzttqq] (-15.839269866743738,1.031735275690869)-- (-15.671347223824824,-8.306971811825295);
\draw [line width=2pt,color=zzttqq] (-15.671347223824824,-8.306971811825295)-- (-7.499828326074586,-12.830900080944975);
\draw [line width=2pt,color=zzttqq] (-7.499828326074586,-12.830900080944975)-- (0.5037679287567434,-8.016121262548495);
\draw [line width=2pt,color=zzttqq] (0.5037679287567434,-8.016121262548495)-- (0.33584528583783135,1.3225858249676747);
\begin{scriptsize}
\draw [fill=black] (0,20) circle (10pt);
\draw [fill=black] (-8.003596254831324,15.185221181603516) circle (10pt);
\draw[color=black] (-7.166370005776736,17.435266725937694) node {\small $l$};
\draw[color=black] (-4.166370005776736,13.435266725937694) node {\small $f(l) = -a_6$};

\draw[color=zzttqq] (-8.840822503885911,11.7839895448193) node {\small $a_6$};
\draw[color=zzttqq] (-3.8174650095583846,2.411727054273536) node {\small $a_1$};
\draw[color=zzttqq] (6.019943416833021,3.5163803354053575) node {\small $a_2$};
\draw [fill=uuuuuu] (-7.8356736119124095,5.846514094087353) circle (10pt);
\draw[color=uuuuuu] (-8.957063443513089,7.121124705205527) node {\small $i$};
\draw[color=uuuuuu] (-3.957063443513089,6.121124705205527) node {\small $f(i) = 0$};
\draw [fill=uuuuuu] (0.33584528583783135,1.3225858249676747) circle (10pt);
\draw[color=uuuuuu] (1.2058924847691416,3.6210336165371797) node { \small $j$};
\draw[color=uuuuuu] (4.2058924847691416,0.6210336165371797) node { \small $f(j) = a_1$};
\draw [fill=uuuuuu] (8.339441540669156,6.137364643364156) circle (10pt);
\draw[color=uuuuuu] (7.459541850787725,7.035084548600993) node {\small $k$};
\draw[color=uuuuuu] (12.159541850787725,4.435084548600993) node {\small $f(k) = a_1 + a_2$};
\draw [fill=uuuuuu] (8.171518897750246,15.476071730880317) circle (10pt);
\draw[color=zzttqq] (-11.457154532181498,2.132712363700909) node {\small $a'_6$};
\draw[color=zzttqq] (1.5989581074056113,-3.030243564581211) node {\small $a'_2$};
\draw [fill=uuuuuu] (-15.839269866743738,1.031735275690869) circle (10pt);
\draw[color=uuuuuu] (-14.596752966136202,3.3070737731417137) node {\small $l'$};
\draw[color=uuuuuu] (-11.596752966136202,-1.3070737731417137) node {\small $f(l') = -a'_6$};
\draw [fill=uuuuuu] (-15.671347223824824,-8.306971811825295) circle (10pt);
\draw [fill=uuuuuu] (-7.499828326074586,-12.830900080944975) circle (10pt);
\draw [fill=uuuuuu] (0.5037679287567434,-8.016121262548495) circle (10pt);
\draw[color=uuuuuu] (-0.429158890428259,-6.5977616853268055) node {\small $k'$};
\draw[color=uuuuuu] (4.729158890428259,-9.7977616853268055) node {\small $f(k') = a_1  + a'_2$};
\end{scriptsize}
\end{tikzpicture}
\caption{The 1-Lipschitz function configuration for computing $K^T_\ipj$. The function is not assigned beyond vertices $i,j,k,k',l,l'$. $K^T_\ipj$ only has the interpretation of curvature if the function can be extended to the rest of the graph, and if $K^T_\ipj$ corresponds to the Wasserstein maximum.}
\label{figliphex}
\end{figure}

\begin{lemma}
\label{lemma10}
For any edge $\ipj\in E\lb\Sigma\rb$, we have 
\be 
K^T_\ipj \leq K_\ipj.
\ee
\end{lemma}
\begin{proof}
For an edge $\ipj\in E\lb \Sigma\rb$, we denote by $W^T_\ipj$ the Wasserstein cost obtained by saturating the 1-Lipschitz extremization for $\ipj$ and the edges adjacent to it (see Figure \ref{figliphex}). Since this is the same as extremization on the tree, $K^T_\ipj$ in Eq. \eqref{KTis} is related to $W^T_\ipj$ via the formula 
\be
K^T_\ipj = \lim_{t\to 0} \frac{1}{t}\lb 1 - \frac{W^t_\ipj}{P_\ipj}\rb.
\ee 
Because $W^T_\ipj$ may not correspond to the minimum transportation plan, we have
\be
W_\ipj \leq W^T_\ipj,
\ee
and since $K_\ipj$ is defined as 
\be
K_\ipj = \lim_{t\to0} \frac{1}{t}\lb1 - \frac{W_\ipj}{P_\ipj}\rb,
\ee 
it follows that
\be
K_\ipj \geq K^T_\ipj.
\ee
\end{proof}

\begin{theorem} For the action 
\be
S_{\Sigma} = \sum_{\ipj \in E(\Sigma)} K_\ipj
\ee
we have $S_{\Sigma} \geq S_{T}$. Furthermore, for the strong boundary condition, $S_{\Sigma}$ achieves its minimum when all edges in $E(\Sigma)$ have constant edge length.
\end{theorem}
\begin{proof} The inequality follows directly from Lemma \ref{lemma10}. From Lemma \ref{lemma99}, $S_T$ achieves its minimum for the constant edge length setting on $E(\Sigma)$. We assign the 1-Lipschitz function at the vertices so that the minimum transportation plan is equal to the one used for computing $W^T_\ipj$ (i.e. such that the Lipschitz inequality is saturated, as in Figure \ref{figliphex}). If the strong boundary condition holds, then (by assigning the Lipschitz function so that the Lipschitz inequality is saturated for all edges one and two steps away from $E\lb\pd\Sigma\rb$ it is possible to extend the assignment of the Lipschitz function to the entire graph, $G$ so that the Lipschitz inequality is obeyed everywhere. Then $K_\ipj = K^T_\ipj$ for $\ipj\in E\lb\Sigma\rb$.
\end{proof}

\begin{remark}
For the hexagon lattice graph with constant edge length setting, the first derivative of the action $S_\Sigma$ with respect to the edge length of a particular edge $\ipj$ is discontinuous.
\end{remark}

For triangle and square lattice graphs, by direct computation the constant edge length setting gives action equal to zero. We conjecture the following.

\begin{conj}
For triangle and square lattice graphs, the constant edge length setting gives the maximum action.
\end{conj}

\subsection{Complete graphs: Maximum action edge length setting}
\label{subseccomplete}

In this section we will present some results for the complete graph on $n$ vertices, which we denote $K_n$. Since complete graphs are finite, in this section we will not have to impose boundary conditions or to introduce boundary terms.

In contrast with trees, we will prove that for complete graphs the constant edge length setting is the maximum action setting.

\begin{lemma}
For $i,j,k \in V(K_n) , P_\ipj + P_\jpk \geq P_\ipk$. This is the triangle inequality on graphs.
\end{lemma}
\begin{proof}
This directly follows from the definition of the geodesic distance.
\end{proof}

\begin{defn}
We define the partial cost between vertices $i,j$ as 
\be 
W^p_{i\to j} \coloneqq \lb 1 - t - \frac{P^{-2}_\ipj}{d_i}t\rb P_\ipj
\ee
\end{defn}
$W^p_{i\to j}$ can be understood as a certain part of the transportation or Wasserstein~cost.

\begin{lemma}
\label{lemma12}
We define an action associated to $W^p$ as
\be
S_p \coloneqq \sum_{i \in V(K_n)} \sum_{j\sim i} \lim_{t\to0} \frac{1 - W^p_{i\to j} P^{-1}_\ipj}{2t}.
\ee
Then for a constant edge length setting on a complete graph with $n$ vertices, the action $S_p$ equals
\be 
S_p = \frac{n^2}{2}.
\ee
\end{lemma}
\begin{proof}
Plugging in the definitions, by direct computation we obtain
\be 
S_p = \sum_{i \in V(K_n)} \sum_{j\sim i} \frac{1 - W^p_{i\to j} P^{-1}_\ipj}{2t} = \frac{n^2}{2}.
\ee
\end{proof}

Now let's prove this is the largest possible action. We approach this by proving that other transportation costs on $\ipj$ will be greater than $W^p_{i\to j}$.

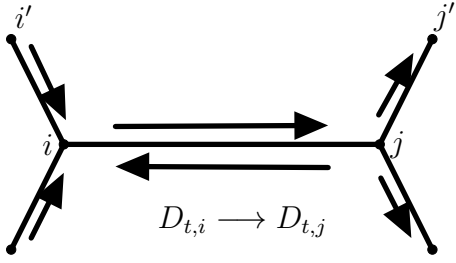
\begin{figure}
\label{transtree}
\centering
\begin{tikzpicture}[line cap=round,line join=round,>=triangle 45,x=1cm,y=1cm,scale=0.7]

\draw [line width=2pt] (-4,2)-- (-3,0);
\draw [line width=2pt] (-3,0)-- (-4,-2);
\draw [line width=2pt] (-3,0)-- (3,0);
\draw [line width=2pt] (3,0)-- (4,2);
\draw [line width=2pt] (3,0)-- (4,-2);
\draw [->,line width=2pt] (-2.02,0.34) -- (2.02,0.36);
\draw [->,line width=2pt] (2.02,-0.44) -- (-2.02,-0.42);
\draw [->,line width=2pt] (-3.62,1.86) -- (-2.96,0.46);
\draw [->,line width=2pt] (-3.62,-1.78) -- (-3,-0.54);
\draw [->,line width=2pt] (2.98,0.56) -- (3.64,1.74);
\draw [->,line width=2pt] (3.02,-0.64) -- (3.6,-1.78);
\draw (-1.42,-0.98) node[anchor=north west] {$D_{t,i} \longrightarrow D_{t,j}$};
\filldraw [black] (-4,2) circle (2.5pt);
\filldraw [black] (-3.76,2.43) node {$i'$};
\filldraw [black] (-4,-2) circle (2.5pt);
\filldraw [black] (-3,0) circle (2.5pt);
\filldraw [black] (-3,0) node [anchor=east]  {$i$};
\draw [fill=black] (3,0) circle (2.5pt);
\draw[color=black] (3,0) node [anchor=west] {$j$};
\draw [fill=black] (4,2) circle (2.5pt);
\draw[color=black] (4.24,2.43) node {$j'$};
\draw [fill=black] (4,-2) circle (2.5pt);
\end{tikzpicture}
\caption{A visualization of the transportation of $D_{t,i}$ to $D_{t,j}$ along edge $\ipj$.}
\label{DiDjtransportvis}
\end{figure}

\begin{lemma}[Local transportation cost lower bound]
\label{lemma13}
For any transportation cost  $W_\ipj$ on $\ipj$, we have
\be 
2W_\ipj \geq W^p_{i\to j} + W^p_{j\to i}.
\ee
\end{lemma}
\begin{proof}
A schematic representation of the Wasserstein transportation distance is in Figure~\ref{DiDjtransportvis}. As in \cite{Gubser:2016htz}, for a small positive parameter $t$, the two probability distributions entering the Wasserstein distance for edge $\ipj$ are defined as
\be
\psi_{t,i}(k) \coloneqq \begin{cases}
1-t \qquad \mrm{if}\ k=i\\
\frac{P_{\ipk}^{-2}}{d_i} t \qquad \mrm{if}\ k\sim i
\end{cases}
\ee
for the probability distribution centered at $i$, and similarly for the one centered at~$j$. Note that in a complete graph any two vertices are neighbors.  When we consider the transportation cost between vertex $i$ and vertex $j$, there must be $W_i\coloneqq 1 - t - P^{-2}_\ipj t/d_j$ amounts of distribution transported out of $i$ and also $j$ should receive $W_j\coloneqq 1 - t - P^{-2}_\ipj t/d_i$. We denote the absolute value of the difference of these two costs by $q$, i.e.
\ba
q&\coloneqq& |W_i - W_j|\\
&=& \left| \frac{1}{d_i} - \frac{1}{d_j} \right| P^{-2}_\ipj  t.
\ea
From the triangle inequality, to achieve the minimum contribution to the Wasserstein distance, the minimum between $W_i$ and $W_j$ must be transported from $i$ to $j$ along the path of geodesic distance $P_\ipj$, giving a contribution to the Wasserstein distance of
\be
\Delta W_1 \coloneqq \lb 1 - t\rb P_\ipj^{-1}  - \max \lb \frac{1}{d_i} , \frac{1}{d_j} \rb P_\ipj^{-1} t.
\ee
Assume wlog that $d_i\leq d_j$, such that $\Delta W_1$ becomes 
\be
\Delta W_1 = \lb 1 - t\rb P_\ipj^{-1} - \frac{P_\ipj^{-1} t}{d_i},
\ee
and $q$ becomes
\be
q = \lb \frac{1}{d_i} - \frac{1}{d_j} \rb P^{-2}_\ipj  t.
\ee
Consider now the probability amount $q$. This amount cannot be moved along the path $P_\ipj$, so it must be moved along the paths $P_\ipk$ or $P_\jpk$ connecting $i$ and $j$ to neighbors $k$ other than $j$ and $i$. Divide $q$ into $q_k$, where $k$ stands for vertices in the graph other than $i$ and $j$. Since $d_i<d_j$, we have $-1/d_i<-1/d_j$ and so $W_j<W_i$, and the probability $q$ must be moved only along the edges out of $i$. This probability goes to the vertices $k$ for which 
\be
\label{eq445}
\frac{P^{-2}_\ipk}{d_i} \leq \frac{P^{-2}_\jpk}{d_j},
\ee
since this probability transfer must decrease the probability distribution centered at $i$ and increase the one centered at $j$. Denote the set of such vertex $k$ by $K$. Then since $d_i<d_j$, from Eq. \eqref{eq445} for these vertices $k\in K$ we must have $P^{-2}_\ipk \leq P^{-2}_\jpk$, that is
\be
\label{eq446}
P_\ipk \geq P_\jpk.
\ee
We thus have that the contribution to the Wasserstein distance of the transfer along edges $\ipk$ with $k\in K$ is
\be
\Delta W_2 \coloneqq \sum_{k \in K} q_k P_\ipk.
\ee
and from Eq. \eqref{eq446} we have
\be
\sum_{k \in K} q_k P_\ipk \geq \sum_{k \in K} q_k P_\jpk.
\ee
We have thus obtained 
\ba
\Delta W_2 &\geq& \sum_{k \in K} \frac{1}{2}q_k \lb P_\ipk + P_\jpk\rb \\
&\geq&  \frac{1}{2}\sum_{k\in K} q_k P_\ipj \\
&=& \frac{1}{2} q P_{\ipj},
\ea
where the last inequality is the triangle inequality. Therefore we have shown
\ba
W_\ipj &\geq& \Delta W_1 + \Delta W_2 \\
&=& \lb 1 - t\rb P_\ipj^{-1} - \lb \frac{1}{d_i} + \frac{1}{d_j} \rb \frac{P_\ipj^{-1} t}{2}.
\ea
This completes the proof.

\end{proof}

We have arrived at the following result.

\begin{theorem}
\label{theorem8}
For the complete graph $K_n$, the constant edge length setting achieves the maximum action, which is equal to $n^2/2$.
\end{theorem}
\begin{proof}
By Lemmas \ref{lemma12} and \ref{lemma13} we have
\ba
S &=& \sum_{\ipj\in E\lb K_n \rb} K_\ipj = \sum_{\ipj\in E\lb K_n \rb} \lim_{t\to 0} \frac{1-W_\ipj P_\ipj^{-1}}{t} \\
&\leq& \sum_{\ipj\in E\lb K_n \rb} \lim_{t\to 0} \frac{1-\lb W^p_{i\to j} + W^p_{j\to i} \rb P_\ipj^{-1}/2}{t}\\
&=& \sum_{i\in V\lb K_n \rb} \sum_{j\sim i} \lim_{t\to 0} \frac{1-W^p_{i\to j } P^{-1}_\ipj }{2t} \\
&=& \frac{n^2}{2}.
\ea
If all edge lengths are equal, then  $d_i=d_j$ for all edges $\ipj\in K_n$, and from the proof of Lemma \ref{lemma13} we have $W_\ipj=2\lb W_{i\to j} + W_{j\to i} \rb$ in this case, so the inequality is saturated.
\end{proof}

\begin{conj}
By direct computation, for complete graphs the action for a perfect matching setting equals $n$, the number of vertices in the graph.  We conjecture that the minimum action for complete graphs $K_n$ with even number of vertices is achieved by the perfect matching setting, in analogy to the tree case.
\end{conj}

We have also conducted an analysis of attainable bounds for some simple finite graphs. The results are encapsulated in the two remarks below.

\begin{remark}
For the finite graph formed by a loop with three vertices (i.e. a triangle with edges $a$, $b$, $c$), the minimum action is $18/5$, given by the edge length setting $l(a) : l(b) : l(c) = 1 : 1 : 2$. The maximum action is $9/2$, obtained by the setting $l(a) : l(b) : l(c) = 1 : 1 : 1$.
\end{remark}

\begin{remark}
  For the finite graph formed by a loop with four vertices (i.e. a square with edges $a$, $b$, $c$, $d$), a numerical analysis suggests that the maximum action is $5$, given by the edge length setting $l(a) = l(b) + 1$, $l(d) = 1$, $l(c) = 0$, and $l(b)$ going to infinity (extremely large compared to $l(d)$ and $l(c)$), and the minimum action is $6 - 2\sqrt{2}$, obtained by $l(a) = 1 + \sqrt{2}$, $l(b) = 1 + \sqrt{2}$, $l(c) = 1$ and $l(d) = 1$.
\end{remark}

For both the triangle and the square the minimum action is given by the degenerate edge length setting, i.e. the edge length setting for which the vertices and edges can be embedded into a line.

\section{Solutions to the equations of motion without boundary condition}
\label{sec555}

In this section we will present solutions to the tree equations of motion (tEoM) \eqref{ac_EOM_tree} for $T_q$, without imposing boundary conditions. 
\begin{defn}
	For a given edge setting $M_q$ for $T_q$, we denote by $M_q/\lambda$ the setting that is obtained by dividing every edge length in $M_q$ by a constant $\lambda\in \mathbb{R}$.
\end{defn}

The following lemma was also given in \cite{Gubser:2016htz}.

\begin{lemma}
	\label{ex_scale}
	For a tree $T_q$, if a setting $M_q$ is a solution to the tEOM, then for any $\lambda\in \mathbb{R}$ and $\lambda \neq 0$, $M_q/\lambda$ is also a solution to the tEOM.
\end{lemma}
\begin{proof}
Immediate by direct computation.

\end{proof}

\subsection{Solutions for $T_1$}
\label{sec_t1}
Even though $q$ should be prime in tree $T_q$, a tree with $q=1$ still offers significant insights. Graphically, $T_1$ is just a line with each vertex being connected to two edges. The tEoM \eqref{ac_EOM_tree} on $\ipj \in T_1$ can be written as 
\begin{equation}
\label{ex_EOM_ij}
P_\ipj^{-1} \lb \frac{(P_\iip^{-1}+P_\ipj^{-1})^2}{(P_\iip^{-2}+P_\ipj^{-2})^2} + \frac{(P_\ipj^{-1}+P_\jjp^{-1})^2}{(P_\ipj^{-2}+P_\jjp^{-2})^2} \rb - \frac{P_\iip^{-1}+P_\ipj^{-1}}{P_\iip^{-2}+P_\ipj^{-2}} - \frac{P_\ipj^{-1}+P_\jjp^{-1}}{P_\ipj^{-2}+P_\jjp^{-2}} = 0.
\end{equation}

Theorem \ref{ex_t1_const} below characterizes the solutions to the tEOM when $q=1$. A version of this theorem may be true even for $q>1$; see Conjectures \ref{ex_conj1} and \ref{ex_conj2}.

\begin{theorem}
	\label{ex_t1_const}
	For any setting of $T_1$ as a solution to the tEOM, if there exists a vertex $i \in T_1$ such that both edges connected to $i$ have the same length, then the setting is a constant solution. 
\end{theorem}

\begin{proof}
	We prove this by induction. Start from an edge denoted $N_0$. Denote the edges on one side of $N_0$ as $N_1$, $N_2$, $N_3,\dots$ in order, and the edges on the other side as $N_{-1}$, $N_{-2}$, $N_{-3},\dots$, as shown in \Cref{ex_pr}. 
	
	First consider the positive side. For the base step, suppose $P_{N_0}=P_{N_1}$. By \Cref{ac_maxmin}, we have $P_{N_2}=P_{N_0}=P_{N_1}$. Similarly, for any non-negative integer $n$, if we assume $P_{N_n}=P_{N_{n+1}}$, by \Cref{ac_maxmin}, we have $P_{N_{n+2}}=P_{N_n}=P_{N_{n+1}}$, which completes the induction step. 
	
	A similar proof follows for the negative side.
\end{proof}

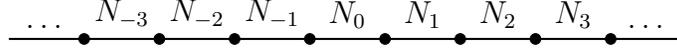
\begin{figure}[htbp!]
	\begin{center}
	\begin{tikzpicture}
	\filldraw[black] (-4,0) circle (2pt) ;
	\filldraw[black] (-3,0) circle (2pt) ;
	\filldraw[black] (-2,0) circle (2pt) ;
	\filldraw[black] (-1,0) circle (2pt) ;
	\filldraw[black] (0,0) circle (2pt) ;
	\filldraw[black] (1,0) circle (2pt) ;
	\filldraw[black] (2,0) circle (2pt) ;
	\filldraw[black] (3,0) circle (2pt) ;
	
	\draw[black,thick] (-5,0) -- (4,0);
	
	\filldraw[black] (-0.5,0) circle (0pt) node[anchor=south] {\normalsize $N_0$};
	\filldraw[black] (-1.5,0) circle (0pt) node[anchor=south] {\normalsize $N_{-1}$};
	\filldraw[black] (-2.5,0) circle (0pt) node[anchor=south] {\normalsize $N_{-2}$};
	\filldraw[black] (-3.5,0) circle (0pt) node[anchor=south] {\normalsize $N_{-3}$};
	\filldraw[black] (-4.5,0) circle (0pt) node[anchor=south] {\normalsize $\dots$};
	\filldraw[black] (0.5,0) circle (0pt) node[anchor=south] {\normalsize $N_{1}$};
	\filldraw[black] (1.5,0) circle (0pt) node[anchor=south] {\normalsize $N_{2}$};
	\filldraw[black] (2.5,0) circle (0pt) node[anchor=south] {\normalsize $N_{3}$};
	\filldraw[black] (3.5,0) circle (0pt) node[anchor=south] {\normalsize $\dots$};
	\end{tikzpicture} 
	\caption{Edge labeling for $q=1$.}
	\label{ex_pr}
	\end{center}
	
\end{figure}

\begin{corollary}
\label{ex_t1_mono}
	For $T_1$, if a setting $M_1$ is a non-constant solution to the tEoM, then its edge lengths are strictly monotonic. 
\end{corollary}
\begin{proof}
	This corollary follows \Cref{ex_t1_const} and \Cref{ac_maxmin}.
\end{proof}

\begin{theorem}
	\label{ex_main_thm1}
	Denote the edges in $T_1$ by $\{N_i\}_{i\in \mathbb{Z}}$, as in \Cref{ex_pr}. For a strictly monotonic setting $M_1$, let $r_k\coloneqq P_{N_k}^{-1}/P_{N_{k-1}}^{-1}$, and suppose wlog $r_0>1$. Then $M_1$ is a solution to the tEoM if and only if for any $k \in \mathbb{Z}$, $P_{N_{k+1}}^{-1}/P_{N_k}^{-1}$ equals either $r_0$ or $(r_0+1)/(r_0-1)$.
\end{theorem}

\begin{proof}
\label{ex_pf_1}
	Since $M_1$ is a solution to the tEoM \eqref{ex_EOM_ij}, for edge $N_1$ we have
\be
P_{N_0}^{-1} \lsb \frac{(P_{N_0}^{-1}+P_{N_1}^{-1})^2}{(P_{N_0}^{-2}+P_{N_1}^{-2})^2} + \frac{(P_{N_0}^{-1}+P_{N_{-1}}^{-1})^2}{(P_{N_0}^{-2}+P_{N_{-1}}^{-2})^2} \rsb - \frac{P_{N_0}^{-1}+P_{N_1}^{-1}}{P_{N_0}^{-2}+P_{N_1}^{-2}} - \frac{P_{N_0}^{-1}+P_{N_{-1}}^{-1}}{P_{N_0}^{-2}+P_{N_{-1}}^{-2}} =0.
\ee
Solving this equation for $r_1$ in terms of $r_0$ we obtain four roots,
\be
r_1 \in \lcb r_0, \frac{r_0+1}{r_0-1}, \frac{1-r_0}{r_0+1},-\frac{1}{r_0} \rcb.
\ee
Since $r_0>1$, only the first two roots give positive edge lengths. Next, solving for $r_2$ in terms of $r_1$ and demanding positivity we obtain
\be
r_2 = r_1\quad \mrm{or} \quad r_2 =\frac{r_1+1}{r_1-1},
\ee
which in terms of $r_0$ is again
\be
r_2 = r_0\quad \mrm{or} \quad r_2 =\frac{r_0+1}{r_0-1}.
\ee
Thus, by induction any $r_k$ must equal $r_0$ or $(r_0+1)/(r_0-1)$, and any sequence of these two ratios gives a monotonic setting.
\end{proof}

Since \Cref{ex_t1_const} and \Cref{ex_main_thm1} cover all cases for any two adjacent edge lengths in a setting, they summarize all types solutions to the tEoM for $T_1$, up to the overall scale factor.

\subsection{Solutions for $T_q$ with odd $q>1$}
This section discuss some solutions to the tEOM on $T_q$ with odd $q>1$. When $q$ is odd, there are an even number of edges connected to each vertex.
\begin{defn}[Half-half setting]
	A half-half setting for $T_q$ with $q$ odd is any setting such that for every vertex $i$ in $T_q$, $(q+1)/2$ edges have the same edge length $\ell_1(i)$, and the other $(q+1)/2$ edges have the same edge length $\ell_2(i)$.
\end{defn}

Note that $\ell_1(i)$ and $\ell_2(i)$ can change from one vertex to another, and by definition a constant solution is a half-half solution.

\begin{theorem}
	\label{ex_HH_const}
	For any half-half setting of $T_q$ that is a solution to the tEoM, if there exists a vertex $i \in V(T_q)$ such that all edges connected to $i$ are equal, then the setting is a constant solution. 
\end{theorem}

\begin{proof}
	Similar to Theorem \ref{ex_t1_const}, this theorem follows immediately from \Cref{ac_maxmin} and induction.

\end{proof}

\begin{theorem}
\label{thmmmm13}
	For $T_q$ with odd q, a non-constant half-half setting is a solution to the tEoM if and only if every strictly monotonic path in the setting is a solution to the tEoM as a setting for $T_1$.
\end{theorem}

We call this type of solution a \textit{half-half solution}, and it can have infinitely many shapes. \Cref{ex_HH} gives an example of a half-half solution for $T_3$. Note that half-half solutions are more general than the solutions presented in Section 5 of \cite{Gubser:2016htz}, since for those solutions the ratio of two distinct edge lengths adjacent to the same vertex is always the same factor $\beta$, whereas the half-half solutions can have multiple values for the ratio.

\begin{proof}
Let's now prove Theorem \ref{thmmmm13}. Suppose $M_q$ is a non-constant half-half setting for $T_q$ and let $n\coloneqq(q+1)/2$. Given any edge $\ij \in E\lb T_q\rb$, $n$ of the edges connected to $i$ have length $P_\ipj$, and denote the length of the other $n$ edges by $P_\iip$. Similarly, we denote the length of the $n$ edges adjacent to $j$ by $P_\jjp$. By Theorem \ref{ex_HH_const}, we have $P_\iip \neq P_\ij$, $P_\jjp \neq P_\ij$ and $P_\iip$, $P_\ij$, $P_\jjp$ are strictly monotonic. The tEOM for $\ipj$ both in $T_q$ and in the $T_1$ given by a path containing $\iip$, $\ipj$, and $\jjp$ is
\be
P_\ipj^{-1} \lb \frac{(P_\iip^{-1}+P_\ipj^{-1})^2}{(P_\iip^{-2}+P_\ipj^{-2})^2} + \frac{(P_\ipj^{-1}+P_\jjp^{-1})^2}{(P_\ipj^{-2}+P_\jjp^{-2})^2} \rb - \frac{P_\iip^{-1}+P_\ipj^{-1}}{P_\iip^{-2}+P_\ipj^{-2}} - \frac{P_\ipj^{-1}+P_\jjp^{-1}}{P_\ipj^{-2}+P_\jjp^{-2}}=0,
\ee
which proves the theorem.
\end{proof}

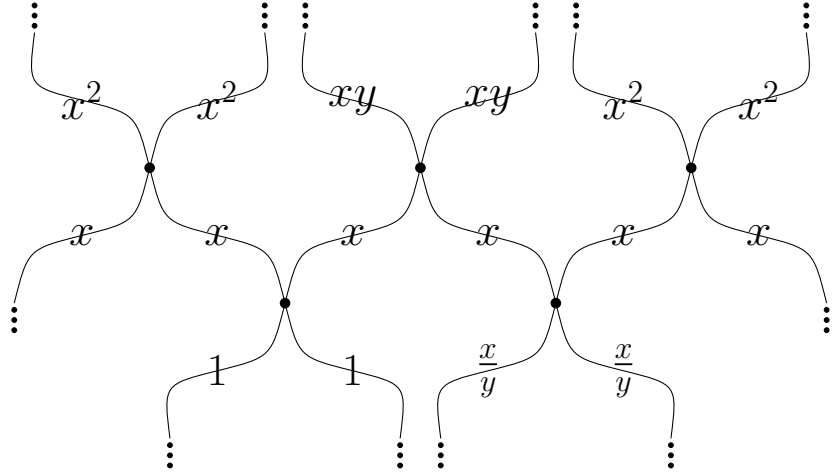
\begin{figure}[htbp!]
	\begin{center}

        \begin{tikzpicture}[scale=0.9]

	\draw (0,1) .. controls (0.2,0.2) ..(1,0) .. controls (1.8,-0.2)  ..(2,-1);
	\draw (0,1) .. controls (-0.2,0.2) ..(-1,0) .. controls (-1.8,-0.2)  ..(-2,-1);

	\draw (4,1) .. controls (4.2,0.2) ..(5,0) .. controls (5.8,-0.2)  ..(6,-1);
	\draw (4,1) .. controls (3.8,0.2) ..(3,0) .. controls (2.2,-0.2)  ..(2,-1);
	
	\draw (-4,1) .. controls (-3.8,0.2) ..(-3,0) .. controls (-2.2,-0.2)  ..(-2,-1);
	\draw (-4,1) .. controls (-4.2,0.2) ..(-5,0)  .. controls (-5.8,-0.2)  ..(-6,-1);
	
	\draw (0,1) .. controls (0.2,1.8) ..(1,2) .. controls (1.8,2.2)  ..(1.7,3);
	\draw (0,1) .. controls (-0.2,1.8) ..(-1,2) .. controls (-1.8,2.2)  ..(-1.7,3);
	
	\draw (4,1) .. controls (4.2,1.8) ..(5,2) .. controls (5.8,2.2)  ..(5.7,3);
	\draw (4,1) .. controls (3.8,1.8) ..(3,2) .. controls (2.2,2.2)  ..(2.3,3);
	
	\draw (-4,1) .. controls (-3.8,1.8) ..(-3,2) .. controls (-2.2,2.2)  ..(-2.3,3);
	\draw (-4,1) .. controls (-4.2,1.8) ..(-5,2)  .. controls (-5.8,2.2)  ..(-5.7,3);
	
	\draw (0.3,-3) .. controls(0.2,-2.2) .. (1,-2) .. controls(1.8,-1.8) .. (2,-1) .. controls(2.2,-1.8) .. (3,-2) .. controls(3.8,-2.2) .. (3.7,-3);
	\draw (-3.7,-3) .. controls(-3.8,-2.2) .. (-3,-2) .. controls(-2.2,-1.8) .. (-2,-1) .. controls(-1.8,-1.8) .. (-1,-2) .. controls(-0.2,-2.2) .. (-0.3,-3);
	
	\filldraw[black] (0,1) circle (2pt) ;
	\filldraw[black] (4,1) circle (2pt);
	\filldraw[black] (-4,1) circle (2pt) ;
	\filldraw[black] (2,-1) circle (2pt) ;
    \filldraw[black] (-2,-1) circle (2pt) ;
    
    \filldraw[black] (1,0) circle (0pt) node[anchor=center] {\Large $x$};
     \filldraw[black] (3,0) circle (0pt) node[anchor=center] {\Large $x$};
     \filldraw[black] (5,0) circle (0pt) node[anchor=center] {\Large $x$};
      \filldraw[black] (-1,0) circle (0pt) node[anchor=center] {\Large $x$};
      \filldraw[black] (-3,0) circle (0pt) node[anchor=center] {\Large $x$};
      \filldraw[black] (-5,0) circle (0pt) node[anchor=center] {\Large $x$};
      
      \filldraw[black] (1,2) circle (0pt) node[anchor=center] {\Large $xy$};
      \filldraw[black] (3,2) circle (0pt) node[anchor=center] {\Large $x^2$};
      \filldraw[black] (5,2) circle (0pt) node[anchor=center] {\Large $x^2$};
      \filldraw[black] (-1,2) circle (0pt) node[anchor=center] {\Large $xy$};
      \filldraw[black] (-3,2) circle (0pt) node[anchor=center] {\Large $x^2$};
      \filldraw[black] (-5,2) circle (0pt) node[anchor=center] {\Large $x^2$};
   
      \filldraw[black] (-1,-2) circle (0pt) node[anchor=center] {\Large $1$};   
      \filldraw[black] (-3,-2) circle (0pt) node[anchor=center] {\Large $1$}; 
      
      \filldraw[black] (1,-2) circle (0pt) node[anchor=center] {\Large $\frac{x}{y}$};   
      \filldraw[black] (3,-2) circle (0pt) node[anchor=center] {\Large $\frac{x}{y}$};   
      
      	\filldraw[black] (6,-1.1) circle (1pt) ;
      	\filldraw[black] (6,-1.25) circle (1pt) ;
      	\filldraw[black] (6,-1.4) circle (1pt) ;
      	
      	\filldraw[black] (-6,-1.1) circle (1pt) ;
      	\filldraw[black] (-6,-1.25) circle (1pt) ;
      	\filldraw[black] (-6,-1.4) circle (1pt) ;
      	
      	\filldraw[black] (-5.7,3.1) circle (1pt) ;
      	\filldraw[black] (-5.7,3.25) circle (1pt) ;
      	\filldraw[black] (-5.7,3.4) circle (1pt) ;
      	
      	\filldraw[black] (-2.3,3.1) circle (1pt) ;
      	\filldraw[black] (-2.3,3.25) circle (1pt) ;
      	\filldraw[black] (-2.3,3.4) circle (1pt) ;
      	
        \filldraw[black] (-1.7,3.1) circle (1pt) ;
      	\filldraw[black] (-1.7,3.25) circle (1pt) ;
      	\filldraw[black] (-1.7,3.4) circle (1pt) ;
      	
      	\filldraw[black] (5.7,3.1) circle (1pt) ;
      	\filldraw[black] (5.7,3.25) circle (1pt) ;
      	\filldraw[black] (5.7,3.4) circle (1pt) ;
      	
      	\filldraw[black] (2.3,3.1) circle (1pt) ;
      	\filldraw[black] (2.3,3.25) circle (1pt) ;
      	\filldraw[black] (2.3,3.4) circle (1pt) ;
      	
      	\filldraw[black] (1.7,3.1) circle (1pt) ;
      	\filldraw[black] (1.7,3.25) circle (1pt) ;
      	\filldraw[black] (1.7,3.4) circle (1pt) ;
      	
      	\filldraw[black] (0.3,-3.1) circle (1pt) ;
      	\filldraw[black] (0.3,-3.25) circle (1pt) ;
      	\filldraw[black] (0.3,- 3.4) circle (1pt) ;
      	
      	\filldraw[black] (-0.3,-3.1) circle (1pt) ;
      	\filldraw[black] (-0.3,-3.25) circle (1pt) ;
      	\filldraw[black] (-0.3,- 3.4) circle (1pt) ;
      	
      	\filldraw[black] (3.7,-3.1) circle (1pt) ;
      	\filldraw[black] (3.7,-3.25) circle (1pt) ;
      	\filldraw[black] (3.7,- 3.4) circle (1pt) ;
      	
      	\filldraw[black] (-3.7,-3.1) circle (1pt) ;
      	\filldraw[black] (-3.7,-3.25) circle (1pt) ;
      	\filldraw[black] (-3.7,- 3.4) circle (1pt) ;
      	
	\end{tikzpicture} 
	\caption{A half-half solution for $T_3$. Note $x \in \mathbb{R}$ and $x>1$, $y=(x+1)/(x-1)$. \label{ex_HH}}
	\end{center}
	
\end{figure}
\begin{defn}
	A geometric half-half solution is a half-half solution in which all strictly monotonic paths are in the same geometric progression. 
\end{defn} 

\begin{remark}
The tree curvature of any geometric half-half solution for $T_q$ is constant and equal to
\be
\label{Kipj5p7}
K_\ipj = \frac{3-q}{1+q}-\frac{2r}{1+r^2},
\ee
where $r$ is the ratio of the progression. The $c_i^2/d_i$ ratio of any vertex in a geometric half-half solution for $T_q$ is constant and equal to
\be
\frac{c_i^2}{d_i} = \frac{(q+1) (1+r)^2}{2 \left(1+r^2\right)}.
\ee
Expression \eqref{Kipj5p7} is always positive for $q=1$. This is consistent with the intuition that the curvature being negative comes from the branching, which does not happen when the graph is a line. Note also that Eq. \eqref{Kipj5p7} does not apply for $q=2$.
\end{remark}

\begin{remark}
\label{thm13}
	There exist solutions to the tEoM that are not half-half solutions. 
\end{remark}
	Such solutions can be constructed by the following algorithm, for tree $T_q$ with $q$~odd:
	\begin{enumerate}
		\item Start from a vertex $i$ such that $m$ edges connected to it have length $1$, another $m$ have length $x$, $s$ edges have length $\alpha$, and another $s$ have length $\alpha y$, where $x$, $y$, $\alpha$ are positive real numbers and $m$, $s$ are non-negative integers such that $2(m+s)=q+1$;
		\item Given a vertex $i'\sim i$, if $P_\iip=1$, let the edges connected to $i'$ have the length of the edges connected to $i$ divided by $x$; if $P_\iip=x$, let the edges connected to $i'$ have the length of the edges connected to $i$ multiplied by $x$; if $P_\iip=\alpha$, let the edges connected to $i'$ have the length of the edges connected to $i$ divided by $y$; if $P_\iip=\alpha y$, let the edges connected to $i'$ have the length of the edges connected to $i$ multiplied by $y$;
		\item Using the rules above, construct the setting iteratively for all the other vertices. 
	\end{enumerate}
	
\begin{figure}[H]
	\begin{center}

        \begin{tikzpicture}
            \draw[black,thick] (-3,0) -- (7,0);
            \draw[black,thick] (-3,2) -- (1,2);
            \draw[black,line width=0.7mm] (-1,4) -- (-1,-2);
            \draw[black,line width=0.7mm] (3,2) -- (3,-2);
            
            \filldraw[black] (-1,0) circle (2pt) ;
            \filldraw[black] (-1,2) circle (2pt) ;
            \filldraw[black] (3,0) circle (2pt) ;
            
            \filldraw[black] (-2,0) circle (0pt) node[anchor=south] {\large $1$};
            \filldraw[black] (1,0) circle (0pt) node[anchor=south] {\large $x$};
            \filldraw[black] (5,0) circle (0pt) node[anchor=south] {\large $x^2$};
            
            \filldraw[black] (-2,2) circle (0pt) node[anchor=south] {\large $y$};
            \filldraw[black] (1,2) circle (0pt) node[anchor=south] {\large $xy$};
            
            \filldraw[black] (-1,-1) circle (0pt) node[anchor=west] {\large $\alpha$};
            \filldraw[black] (-1,1) circle (0pt) node[anchor=west] {\large $\alpha y$};
            \filldraw[black] (-1,3) circle (0pt) node[anchor=west] {\large $\alpha y^2$};
            
            \filldraw[black] (3,-1) circle (0pt) node[anchor=west] {\large $\alpha x$};
            \filldraw[black] (3,1) circle (0pt) node[anchor=west] {\large $\alpha x y$};
            
            \filldraw[black] (-1,4.1) circle (1pt) ;
            \filldraw[black] (-1,4.25) circle (1pt) ;
            \filldraw[black] (-1,4.4) circle (1pt) ;
            
            \filldraw[black] (-1,-2.1) circle (1pt) ;
            \filldraw[black] (-1,-2.25) circle (1pt) ;
            \filldraw[black] (-1,-2.4) circle (1pt) ;
            
            \filldraw[black] (3,2.1) circle (1pt) ;
            \filldraw[black] (3,2.25) circle (1pt) ;
            \filldraw[black] (3,2.4) circle (1pt) ;
            
            \filldraw[black] (3,-2.1) circle (1pt) ;
            \filldraw[black] (3,-2.25) circle (1pt) ;
            \filldraw[black] (3,-2.4) circle (1pt) ;
            
            \filldraw[black] (-3.1,0) circle (1pt) ;
            \filldraw[black] (-3.25,0) circle (1pt) ;
            \filldraw[black] (-3.4,0) circle (1pt) ;
            
            \filldraw[black] (7.1,0) circle (1pt) ;
            \filldraw[black] (7.25,0) circle (1pt) ;
            \filldraw[black] (7.4,0) circle (1pt) ;
            
            \filldraw[black] (-3.1,2) circle (1pt) ;
            \filldraw[black] (-3.25,2) circle (1pt) ;
            \filldraw[black] (-3.4,2) circle (1pt) ;
            
            \filldraw[black] (1.1,2) circle (1pt) ;
            \filldraw[black] (1.25,2) circle (1pt) ;
            \filldraw[black] (1.4,2) circle (1pt) ;
            
            \end{tikzpicture}
	\end{center}
	\caption{\label{ex_general} Building the setting iteratively when $q=3$, according to the algorithm in the proof of Theorem \ref{thm13}. The starting vertex is lower left; moving along the $m$ edges multiplies by $x$ or $x^{-1}$, and along the $s$ edges multiplies by $y$ or $y^{-1}$.}
\end{figure}
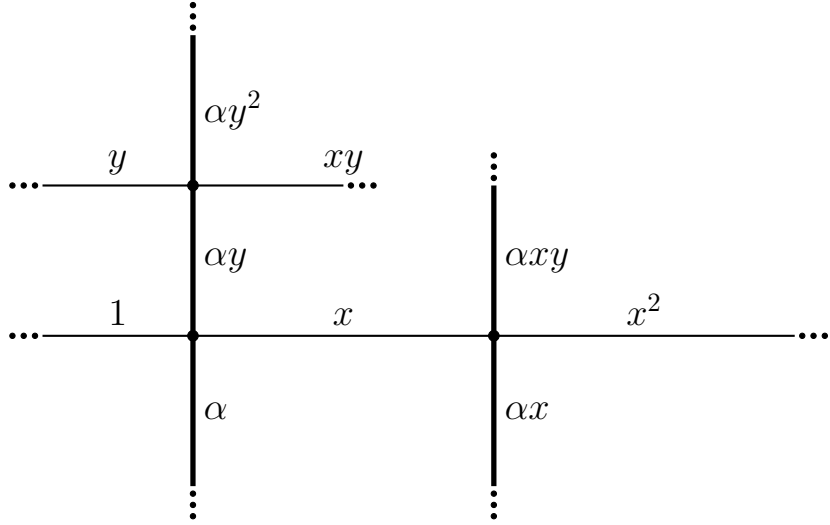
	
	For example, when $m=s=1$, i.e. $q=3$, we obtain the setting shown in \Cref{ex_general}. 
	
	Note that this edge length assignment has $x\leftrightarrow y$, $\alpha\leftrightarrow 1/\alpha$, $s\leftrightarrow m$ symmetry. Because the tEOM are invariant under a constant scaling, a priori there are only two expressions for the equations of motion in terms of the variables $\alpha$, $x$, and $y$, namely on the $m$ edges and on the $s$ edges. However, it turns out for this particular choice of edge lengths that the two equations of motion coincide up to an overall multiplicative factor, and furthermore the factor in the equation of motion which gives nontrivial positive solutions is independent of $m$ and $s$. The equation of motion therefore is equivalent to
	\be
	\alpha  y (y+1)\left(x^2+1\right) = x (x+1)\left(y^2+1\right),
	\ee
	which is solved by
	\be
	\label{eq510}
	x=\frac{y^2+1\pm\sqrt{\lb y^2+1\rb^2-4 \alpha  y (y+1) \lsb \alpha y +(\alpha -1) y^2-1\rsb}}{2 y \lb\alpha +(\alpha-1) y\rb -2}.
	\ee
	Eq. \eqref{eq510} admits positive solutions for $x$ and $y$. For instance, choosing $\alpha=1/4$ and $y=3$ gives
    \be
    x=\frac{1}{7} \left(\sqrt{46}-5\right),
    \ee
    which is positive.

\begin{remark}
In general, the solutions to the tree equations of motions can contain more than two geometric progressions. This adds many possibilities to the forms the solutions can take.
\end{remark}

We conclude by giving two conjectures on the existence of solutions. The intuition behind these conjectures is that if there exists a vertex $i$ such that all edges around it have equal lengths, then evolving the tEOM away from this vertex will lead to edge lengths that are no longer positive after a finite number of steps, unless the edge lengths are all constant.

\begin{conj}
	\label{ex_conj1}
	For $T_q$, if a setting $M_q$ is a solution to the tEoM with positive edge lengths for all edges and there exists a vertex $i \in T_q$ such that all edges connected to $i$ are equal, then $M_q$ is a constant solution.
\end{conj}

\begin{conj}
   \label{ex_conj2}
	There exists no non-constant solution to the tEoM for $T_2$ with all edge lengths positive.
\end{conj}

\end{document}